\documentclass[10pt,letterpaper,conference]{ieeeconf}
\IEEEoverridecommandlockouts                              

\usepackage{graphicx,soul} 
\usepackage{epsfig,color} 
\usepackage{amsmath,algorithm,algorithmic,cite} 
\usepackage{algorithmic}
\usepackage[algo2e]{algorithm2e}
\usepackage{amssymb,mathrsfs}  
\usepackage{theorem}
\usepackage{subfigure}
\usepackage[breaklinks]{hyperref}
\hypersetup{pdfborder={0 0 0},
            breaklinks=true,%
            bookmarks=true,%
            bookmarksnumbered=true,%
            colorlinks=true,%
            citecolor=red,%
            linkcolor=blue,%
}

 \newtheorem{theorem}{Theorem}[section]

\newtheorem{lemma}[theorem]{Lemma}

\newtheorem{definition}{Definition}

\def\1{\mathbf{1}}

\newcommand{\subscr}[2]{{#1}_{\textup{#2}}}

\newcommand{\abs}[1]{|{#1}|}
\DeclareMathOperator*{\argmax}{arg\,max}

\title{Competitive Perimeter Defense of Conical Environments}

\author{Shivam Bajaj$^{1}$, Eric Torng$^{2}$, Shaunak D. Bopardikar$^{1}$, \\ Alexander Von Moll$^{3}$, Isaac Weintraub$^{3}$, Eloy Garcia$^{3}$,  David W. Casbeer$^{3}$
\thanks{$^{1}$S. Bajaj and S. D. Bopardikar are with the Department of Electrical and Computer Engineering, Michigan State University.
Email:\texttt{ bajajshi@msu.edu} (Shivam Bajaj)}
\thanks{$^{2}$E. Torng is with the Department of Computer Science and Engineering, Michigan State University.}
\thanks{$^{3}$A. Von Moll, I. Weintraub, E. Garcia and D. Casbeer are with Control Science
Center, Air Force Research Laboratory.}
\thanks{This research was supported in part by the Air Force Office of Scientific Research Summer Faculty Fellowship Program, Contract Numbers FA8750-15-3-6003, FA9550-15-0001 and FA9550-20-F-0005 and in part by NSF Award ECCS-2030556. Approved for public release: distribution unlimited, case number: AFRL-2021-3011.}
}
\begin{document}
\maketitle
\thispagestyle{empty}
\pagestyle{empty}

\begin{abstract}
We consider a perimeter defense problem in a planar conical environment in which a single vehicle, having a finite capture radius, aims to defend a concentric perimeter from mobile intruders. The intruders are arbitrarily released at the circumference of the environment and they move radially toward the perimeter with fixed speed. We present a competitive analysis approach to this problem by measuring the performance of multiple online algorithms for the vehicle against \emph{arbitrary} inputs, relative to an optimal offline algorithm that has information about entire input sequence in advance. 
In particular, we establish two necessary conditions on the parameter space to guarantee (i) finite competitiveness of any algorithm and (ii) a competitive ratio of at least $2$ for any algorithm.
We then design and analyze three online algorithms and characterize parameter regimes in which they have finite competitive ratios. Specifically, our first two algorithms are provably $1$, and $2$-competitive, respectively, whereas our third algorithm exhibits different competitive ratios in different regimes of problem parameters. Finally, we provide a numerical plot in the parameter space to reveal additional insights into the relative performance of our algorithms. 
\end{abstract}
\section{Introduction}
This work considers a perimeter defense problem in a conical environment in which a single mobile vehicle seeks to intercept mobile intruders before they enter a specified region (referred to as the perimeter). This scenario arises when a UAV is required to tag (or relay critical information to) the intruders (targets) before they reach a specific region of interest. The intruders are generated at the boundary of the environment and move radially inwards with fixed speed toward the perimeter. The vehicle, which has a finite capture radius, moves with bounded speed (greater than that of the intruders) with the aim of \textit{capturing} as many intruders as possible before they reach the perimeter. This is an online problem as the number and the arrival location of intruders is gradually revealed over time.

Most prior works in the area of perimeter defense have either focused on determining optimal strategies for a small number of agents or considered a stochastic arrival process for the intruders~\cite{von2020multiple,macharet2020adaptive,ShivamDVR2019}. Although these studies provide valuable insights, they essentially ignore the \emph{worst-case performance} where the intruders might coordinate their actions to overwhelm the defense \cite{von2021turret}.

In this work, we adopt a competitive analysis technique~\cite{sleator1985amortized} to assess online vehicle motion planning algorithms. In competitive analysis, we measure the performance of an online algorithm, $A$,
using the concept of \emph{competitive ratio}, i.e., the ratio of an optimal offline algorithm's performance divided by algorithm $A$'s performance for a worst-case input instance. Algorithm $A$ is $c$-competitive if its competitive ratio is no larger than $c$ which means its performance is guaranteed to be within a factor, $c$, of the optimal for all input instances. In this work, the performance of an algorithm, either online or offline, is measured by the fraction of intruders captured.


A related area of research is vehicle routing with new inputs arriving over time. Introduced on graphs in~\cite{psaraftis1988dynamic}, a typical approach requires that the vehicle routes be re-planned as new information is revealed over time.
We refer the reader to \cite{bullo2011dynamic} and the references therein for a review of this literature. 
In most of the vehicle routing problems, the input (known as \emph{demands}) are static, and so, the problem is to find the shortest route through the demands in order to minimize (maximize) the cost (reward); examples of such metrics would be the total service time or the number of inputs \emph{serviced}. In perimeter defense scenarios, the input (intruders) are not static. Instead, they are moving towards a specified region, making this problem more challenging than the former. In our previous works, we introduced perimeter defense problems in circular and rectangular environments with stochastically generated input, \cite{ShivamDVR2019,Smith2009translating}. The key distinction of our current work from these past works is the characterization of \emph{competitiveness} for the worst-case inputs, as opposed to the average-case.

Perimeter defense problems were first introduced for a single vehicle and a single intruder in~\cite{isaacs1965differential}. Since then, perimeter defense has been mostly formulated as a pursuit-evasion differential game. The multiplayer setting for the same has been studied extensively as a reach-avoid game in which the aim is to design control policies for the intruders and the defenders \cite{chen2016multiplayer,garcia2019strategies,davydov2020pursuer}. 
A typical approach requires computing solutions to the Hamilton-Jacobi-Bellman-Isaacs equation, which is generally only suitable for low dimensional state spaces and in simple environments \cite{margellos2011hamilton,chen2014path}. Recent works include \cite{yan2019matchingbased, shishika2019perimeter,velhal2021decentralized,lee2021guarding}. Authors in \cite{yan2019matchingbased} propose a receding horizon strategy based on maximum matching,  \cite{ shishika2019perimeter,velhal2021decentralized} consider a scenario wherein the defenders are constrained to be on the perimeter and \cite{lee2021guarding} extends the reach avoid game to $n$-dimensional Euclidean spaces.
Previously, we introduced a perimeter defense problem for linear environments based on the use of competitive analysis \cite{bajaj2021competitive}. The key distinction of our current work from the past work is the geometry of the environment which yields novel results in terms of optimally placing the vehicle, role of capture radius and additional conditions to guarantee competitiveness of the algorithms.

The general contribution of this paper is that we consider a conical environment of unit radius and angle $2\theta$ in which arbitrary number of intruders are released at the circumference of the environment at arbitrary time instances. Upon release, the intruders move radially inwards with fixed speed $v<1$ with the aim of reaching a conical perimeter of radius $\rho<1$ and angle $2\theta$. A single vehicle having a finite capture radius $r$, moves with maximum speed of unity with an aim to capture the intruders. Our main contributions are as follows. We first establish two necessary conditions in the parameter space for achieving a $c$-competitive algorithm with a finite $c$. Specifically, we characterize the parameter regime in which no online algorithm is $c$-competitive and a parameter regime in which no algorithm can be better than $2$-competitive. Next, we design and analyze three classes of algorithms and establish their competitiveness. Specifically, we identify parameter regimes in which the first two algorithms are provably $1$ and $2$-competitive, respectively, and the third algorithm has a finite competitive ratio that varies with the problem parameters ($r,\rho,\theta$).

This paper is organized as follows. In section \ref{sec:Problem}, we formally describe our problem and define the competitive ratio for online algorithms. Section \ref{sec:Fundamental_limits} establishes two necessary conditions; first on achieving a finite competitive ratio and second on achieving at best a competitive ratio of $2$. In section \ref{sec:Algorithms}, we design and analyze three algorithms and establish their competitive ratios, section \ref{sec:Results} provides additional insights through numerous parameter space plots and finally, section \ref{sec:Conclusion} summarizes this work and outlines directions for future works.

\begin{figure}[t]
    \centering
    \includegraphics[scale=0.5]{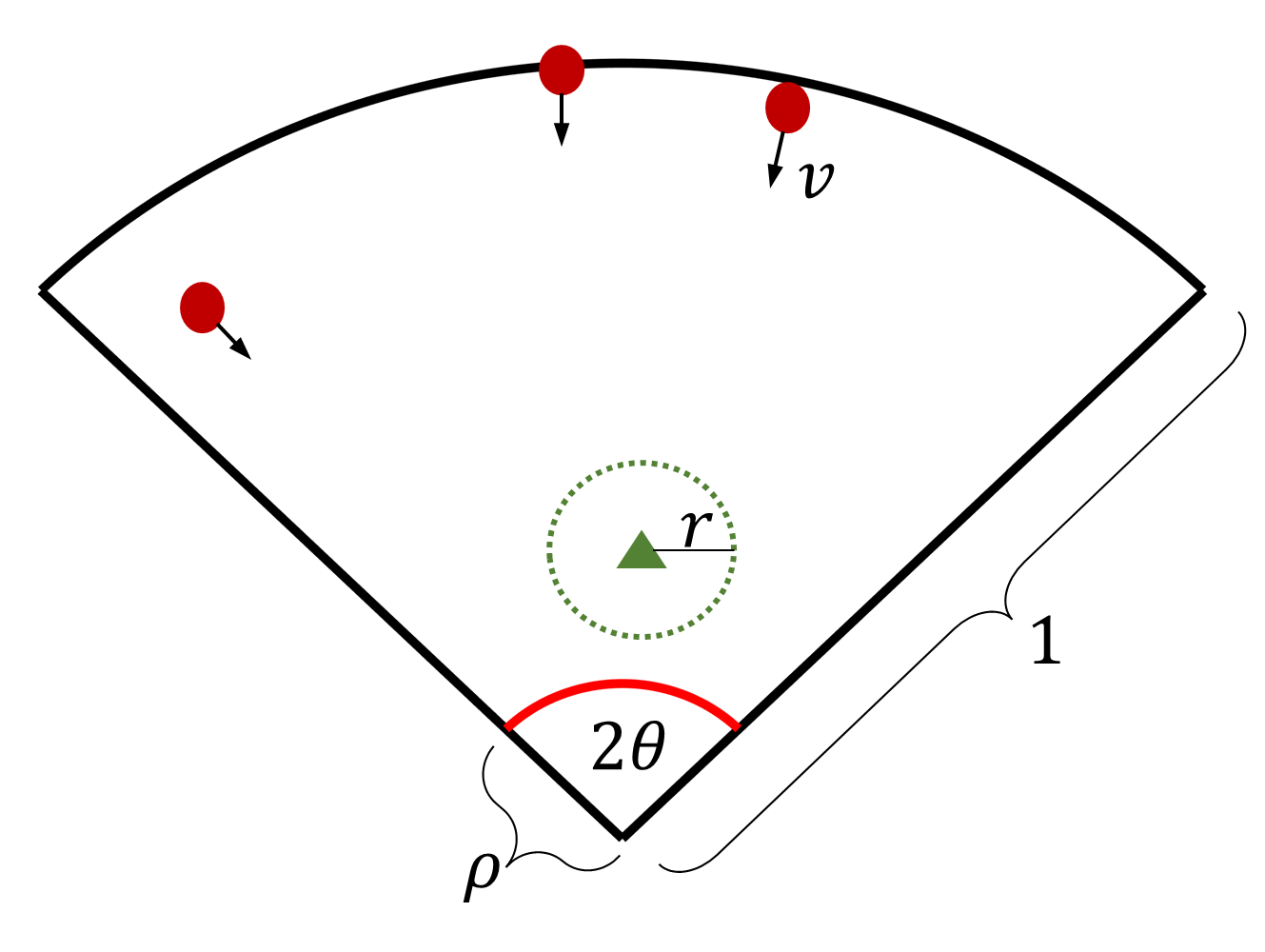}
    \caption{\small Problem Description. The dark red circles represent the intruders, and the vehicle and its capture circle are represented by the green triangle and green circle, respectively. The perimeter is denoted by the red curve.}
    \label{fig:Problem_Desc}
\end{figure}

\section{Problem Description}\label{sec:Problem}
Consider a conical environment of $\mathcal{E}(\theta)=\{(y,\alpha) \, : \, 0 < y \leq 1, -\theta\leq\alpha\leq \theta \}$ which contains a conical region (referred to as perimeter) $\mathcal{R}(\rho,\theta)=\{(z,\alpha) \, : \, 0 < z \leq \rho<1, -\theta\leq\alpha\leq \theta \}$ (Fig. \ref{fig:Problem_Desc}). Intruders arrive over time at the circumference of the environment, i.e., $y=1$ and move radially inwards with a fixed speed $v$ towards the origin in order to breach the perimeter. The defense consists of a single vehicle with motion modeled as a first order integrator with maximum speed of unity and a finite capture radius $r<\rho$ \footnote{If $r\geq \rho$, then the problem is trivial as an algorithm that positions the vehicle at the origin can capture all intruders.}. A \emph{capture circle} is defined as a circle of radius $r$, centered at the vehicle's location. An intruder is \emph{captured} and subsequently removed from $\mathcal{E}(\theta)$ if it lies within or on the capture circle.  
An intruder is said to be \emph{lost} if the intruder reaches the perimeter without being captured by the vehicle.

A \emph{problem instance} $\mathcal{P}(\theta,\rho,v,r)$ is characterized by four parameters: the speed of the intruders, $v<1$, the perimeter's radius $0<\rho<1$, the angle that defines the size of the environment as well as the perimeter, $0<\theta\leq \pi$ and, the capture radius $r<\rho$.
An input instance $\mathcal{I}$ is a set of tuples consisting of time instant $t\leq T$, where $T$ denotes the final time instant, the number of intruders $N(t)$ that are released at time instant $t$, and the arrival location of each of the $N(t)$ intruders. Formally,
$\mathcal{I}=\{t,N(t),\{ (1,\alpha_1),(1,\alpha_2),\dots,(1,\alpha_{N(t)}) \}\}_{t=0}^T$, for any $\alpha_l \in [-\theta,\theta]$ where $1\leq l \leq N(t)$. 



We now formally define an online algorithm.

\textit{Online Algorithm:} An online algorithm is a map $\mathcal{A}: I(t)\rightarrow \mathcal{B}_1$, where $\mathcal{B}_1$ denotes a unit ball. In other words, $\mathcal{A}$ assigns a velocity in the plane with at most unit magnitude to the vehicle as a function of the input $I(t)\subset \mathcal{I}$ revealed until time $t$, yielding the kinematic model, $\mathbf{\dot{x}}(t)=\mathcal{A}(I(t))$, where $\mathbf{x}$ denotes the vehicle's polar coordinates.

An optimal \emph{offline algorithm} is a non-causal algorithm which computes the velocity of the vehicle at any time $t$ as a function of the entire input instance $\mathcal{I}$; that is, it knows in advance when, where, and how many intruders will arrive. 

\begin{definition}[Competitive Ratio]
Given a problem instance $\mathcal{P}(\theta,\rho,r,v)$, an input instance $\mathcal{I}$, and an online algorithm $A$, let $A(\mathcal{I})$ denote the the number of intruders captured by the vehicle when using algorithm A on input instance $\mathcal{I}$. Let $\mathcal{O}$ denote the optimal offline algorithm that maximizes the number of intruders captured out of input instance $\mathcal{I}$. Then, the competitive ratio of $A$ on $\mathcal{I}$ is defined as $c_A(\mathcal{I}) = \tfrac{\mathcal{O}(\mathcal{I})}{A(\mathcal{I})} \ge 1$, and the competitive ratio of $A$ for the problem instance $\mathcal{P}$ is $c_A(\mathcal{P}) = \sup_{\mathcal{I}} c_A(\mathcal{I})$. Finally, the competitive ratio for the problem instance $\mathcal{P}$ is $c(\mathcal{P}) = \inf_A c_A(\mathcal{P})$.
An algorithm is $c$-competitive for the problem instance $\mathcal{P}(\theta,\rho,r,v)$ if $c_A(\mathcal{P}) \leq c$, where $c\geq 1$ is a constant.
\end{definition}

\textit{Problem Statement:} The aim is to establish fundamental guarantees and to design $c$-competitive algorithms for the vehicle with minimum $c$.\\

In light of Lemma 1 in \cite{bajaj2021competitive}, we restrict our attention to extreme speed algorithms that move the vehicle with maximum speed or keep it stationary.



\section{Fundamental Limit}\label{sec:Fundamental_limits}
Before we establish necessary conditions in the space of problem parameters $(\theta,v,r,\rho)$
we provide two properties based on geometry of the environment.
\begin{lemma}\label{lem:capture_rad_lim}
For a problem instance $\mathcal{P}$ with $\theta < \frac{\pi}{4}$, all intruders can be captured if
$r \geq \rho\tan(\theta)$ by positioning the vehicle at $\left(\tfrac{\rho}{\cos(\theta)},0\right)$. 
\end{lemma}
\begin{proof}
Let $(x,\alpha)\in\mathcal{E}(\theta)$ denote a location for the vehicle such that when the vehicle is positioned at $(x,\alpha)$, the capture circle contains the entire circumference of the perimeter.
Note that in order to ensure that the capture circle contains the entire perimeter,  location $(x,\alpha)$ must be equidistant from the points $(\rho,\theta)$ and $(\rho,-\theta)$, i.e., the position of the vehicle $(x,\alpha)$ must lie on the angle bisector that bisects angle $2\theta$ of the environment, implying that $\alpha=0$. Further, the minimum capture radius $r$ must be a value such that the capture circle is tangent to the sector of radius $\rho$ and angle $2\theta$.
Using the property that the radius of a circle is perpendicular to its tangent and applying trigonometric definitions, we obtain
    $r = \rho\tan(\theta)$ and 
    $x = \tfrac{\rho}{\cos(\theta)}$.
This concludes the proof.
\end{proof}

\begin{figure}[t]
    \centering
    \includegraphics[scale=0.5]{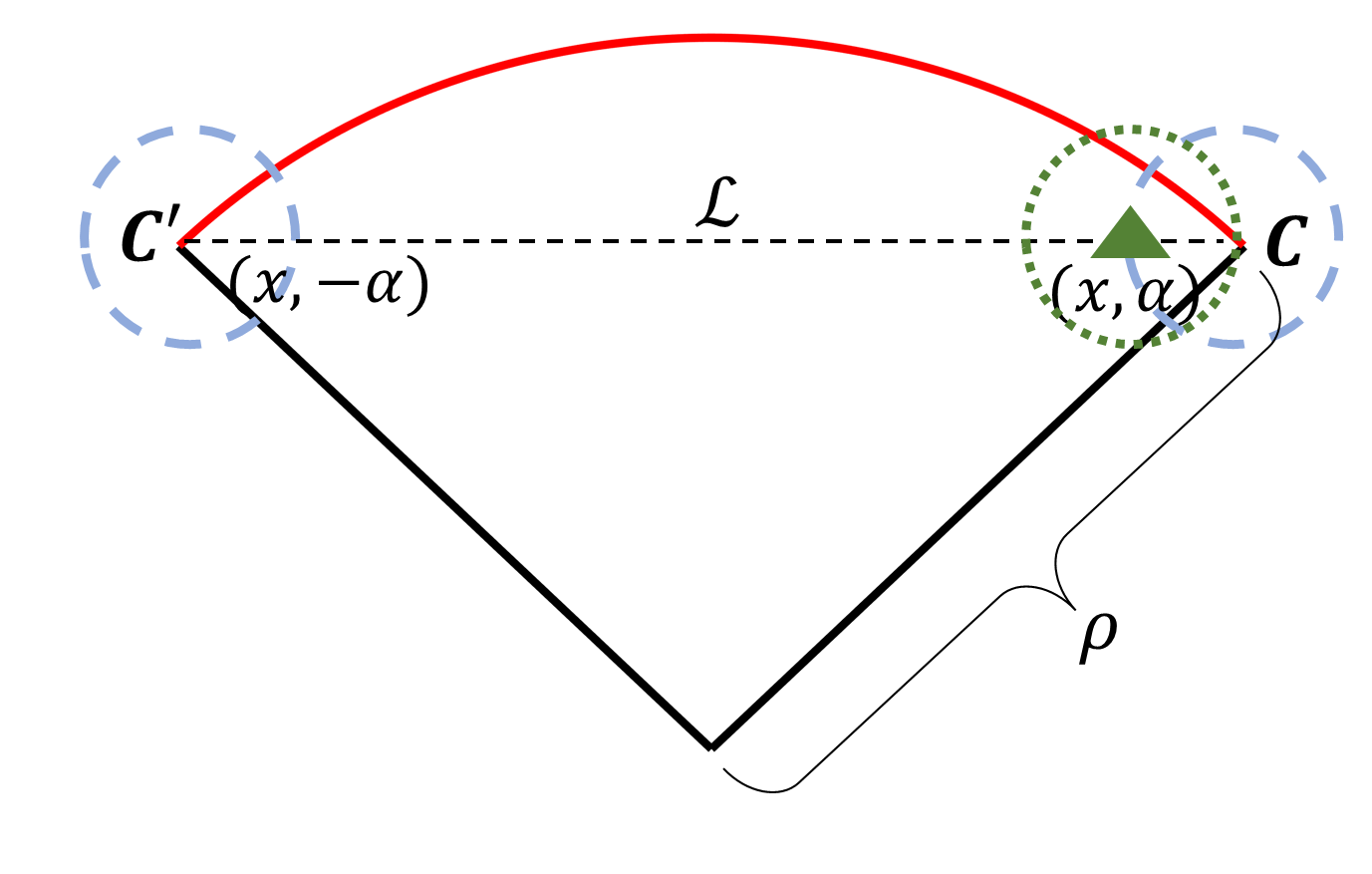}
    \caption{\small Description of proof of Lemma \ref{lem:min_time}. The blue dashed circles $\mathcal{C}$ and $\mathcal{C'}$ are centered at $(\rho,\theta)$ and $(\rho,-\theta)$ respectively. The vehicle, denoted by the green triangle, is located at $(x,\alpha)$. The line $\mathcal{L}$ is denoted by the black dashed line.}
    \label{fig:min_time}
\end{figure}
The next result characterizes the minimum time required by the vehicle to move from one end of the perimeter to the other. 
\begin{lemma}\label{lem:min_time}
The minimum time required by the vehicle to move from a location such that the capture circle contains one end of the perimeter, $(\rho,\theta)$, to a location such that the capture circle contains the opposite end of the perimeter, $(\rho,-\theta)$, is $2(\rho\sin(\theta)-r)$ if $\theta< \frac{\pi}{2}$ and $2(\rho-r)$ otherwise.
\end{lemma}
\begin{proof}
Consider two circles $\mathcal{C}$ and $\mathcal{C}'$, each of radius $r$ and with centers coinciding with the points $(\rho,\theta)$ and $(\rho,-\theta)$, respectively (Fig. \ref{fig:min_time}). 
Observe that, if the vehicle is located at any point within the intersection of the circle $\mathcal{C}$ (resp. $\mathcal{C}'$) and the perimeter, then the location $(\rho,\theta)$ (resp. $(\rho,-\theta)$) will be contained within the capture circle. 

Let $(x,\alpha)\in \mathcal{E}(\theta)$ (resp. $(x,-\alpha)$) denote a point on the circumference of circle $\mathcal{C}$ (resp. $\mathcal{C}'$). Our aim is to determine the closest possible pair of locations $(x,\alpha)$ and $(x,-\alpha)$; that is, the pair of positions that minimizes the distance the vehicle needs to travel to go from one point to the other.
We consider two cases: (i) $\theta\le\tfrac{\pi}{2}$ and (ii) $\theta>\tfrac{\pi}{2}$. 

For case (i), finding the closest pair of points corresponds to determining the shortest distance along the line $\mathcal{L}$ between the two circles $\mathcal{C}$ and $\mathcal{C}'$. The shortest line that connects any two non intersecting circles must pass through the center of the circles, so $\mathcal{L}$ must pass through $(\rho,\theta)$ and $(\rho,-\theta)$, and
the points $(x,\alpha)$ and $(x,-\alpha)$ are where $\mathcal{L}$ intersects $\mathcal{C}$ and $\mathcal{C}'$, respectively.
We compute the distance between the two points as follows. We first find the angle bisector which bisects the angle $2\theta$ of the environment. Note that the angle bisector is also the perpendicular bisector of line $\mathcal{L}$ as the triangle formed by joining points $(\rho,\theta)$, $(\rho,-\theta)$ and the origin is an isosceles triangle. By constructing a triangle that joins the points $(x,\alpha)$, origin, and the midpoint of line $\mathcal{L}$ and using  trigonometric identities, we determine that $x=\frac{\rho\sin(\theta)-r}{\sin(\alpha)}$ and $\alpha = \tan^{-1}(\frac{\rho\sin(\theta)-r}{\rho\cos(\theta)})$. From geometry, it follows that the length of the line segment joining the points $(x,\alpha)$ and $(x,-\alpha)$ is $2(\rho\sin(\theta)-r)$.
Note that when $\theta = \tfrac{\pi}{2}$, $\mathcal{L}$ runs through the origin simplifying this expression to $2(\rho-r)$ matching the expression from case (ii).

For case (ii), the corresponding line $\mathcal{L}$ is not contained in the environment which means the vehicle cannot travel on $\mathcal{L}$ to move from $(x,\alpha)$ to $(x,-\alpha)$. Instead, the shortest path is to  first move to the origin from $(x,\alpha)$ and then to the location $(x,-\alpha)$ from the origin. This gives us $x=\rho-r$ and $\alpha=\theta$ and a minimum distance of $2(\rho -r)$.
This concludes the proof.
\end{proof}

We now present our first necessary condition on the problem parameters for a finite $c(\mathcal{P})$.
\begin{theorem}\label{thm:no_c}
\textbf{(Necessary condition for finite $c(\mathcal{P})$)}
For any problem instance $\mathcal{P}(v,\rho,r,\theta)$ with parameters satisfying 
\begin{align*} 
    2(\rho\sin(\theta)-r)>\frac{1-\rho}{v}, &\text{ if } \theta < \frac{\pi}{2}, \\
    2(\rho-r)>\frac{1-\rho}{v}, &\text{ if } \theta \geq \frac{\pi}{2},
\end{align*}
there does not exist a $c$-competitive algorithm for any constant $c$ and no algorithm, either online or offline, can capture all intruders.
\end{theorem}
\begin{proof}
In this proof, we first construct an input sequence and then determine the number of intruders captured in that input sequence by any online algorithm. Finally, we compare the performance with the performance of an optimal offline algorithm to establish the result.

Consider an online algorithm $\mathcal{A}$ and an optimal offline algorithm $\mathcal{O}$. For both algorithms, assume that the vehicle starts at the origin at time $0$.
The input instance starts at time instant $1$ with a \emph{stream} of intruders, i.e., a single intruder being released every $\frac{1-\rho}{v}$ time units apart, at location $(1,\theta)$. 
If $\mathcal{A}$ never captures any stream intruders, the stream never ends meaning the algorithm $\mathcal{A}$ will not be $c$-competitive for any constant $c\geq 1$, and the first result follows as the optimal offline algorithm can move to $(\rho,\theta)$ and capture all the stream intruders.
We thus assume $\mathcal{A}$ does capture at least one stream intruder, say the $i^{\text{th}}$ one, at time $t$. The input instance ends with the release of a burst of $c+1$ intruders that arrive at location $(1,-\theta)$ at the same time instant $t$.

We now identify how many intruders $\mathcal{A}$ can capture. First, it cannot capture stream intruders 1 through $i-1$ because the stream intruders arrive $\frac{1-\rho}{v}$ time units apart meaning the previous intruder reaches the perimeter and thus is lost just as the next stream intruder arrives.
We now show that the vehicle cannot capture any of the $c+1$ burst intruders.
At time $t$, the vehicle must be at most $r$ distance away from the $i^{\text{th}}$ stream intruder in order to capture it. Likewise, it has only $\frac{1-\rho}{v}$ time  to move to capture the $c+1$ burst intruders that arrived at time $t$. From Lemma \ref{lem:min_time} and our given conditions, $2(\rho\sin(\theta)-r)>\frac{1-\rho}{v}$ (resp. $2(\rho-r)>\frac{1-\rho}{v}$) for $\theta<\frac{\pi}{2}$ (resp. $\theta\geq \frac{\pi}{2}$), the vehicle is guaranteed to not capture the burst intruders.

On the other hand, the optimal offline algorithm
$\mathcal{O}$ can move the vehicle to location $(x,\alpha)$, as defined in Lemma \ref{lem:min_time}, until the first $i-1$ intruders have been captured and then move the vehicle to $(x,-\alpha)$ capturing the burst intruders, losing only the $i^{th}$ intruder. This concludes the proof.
\end{proof}

\begin{figure}
    \centering
    \includegraphics[scale=0.35]{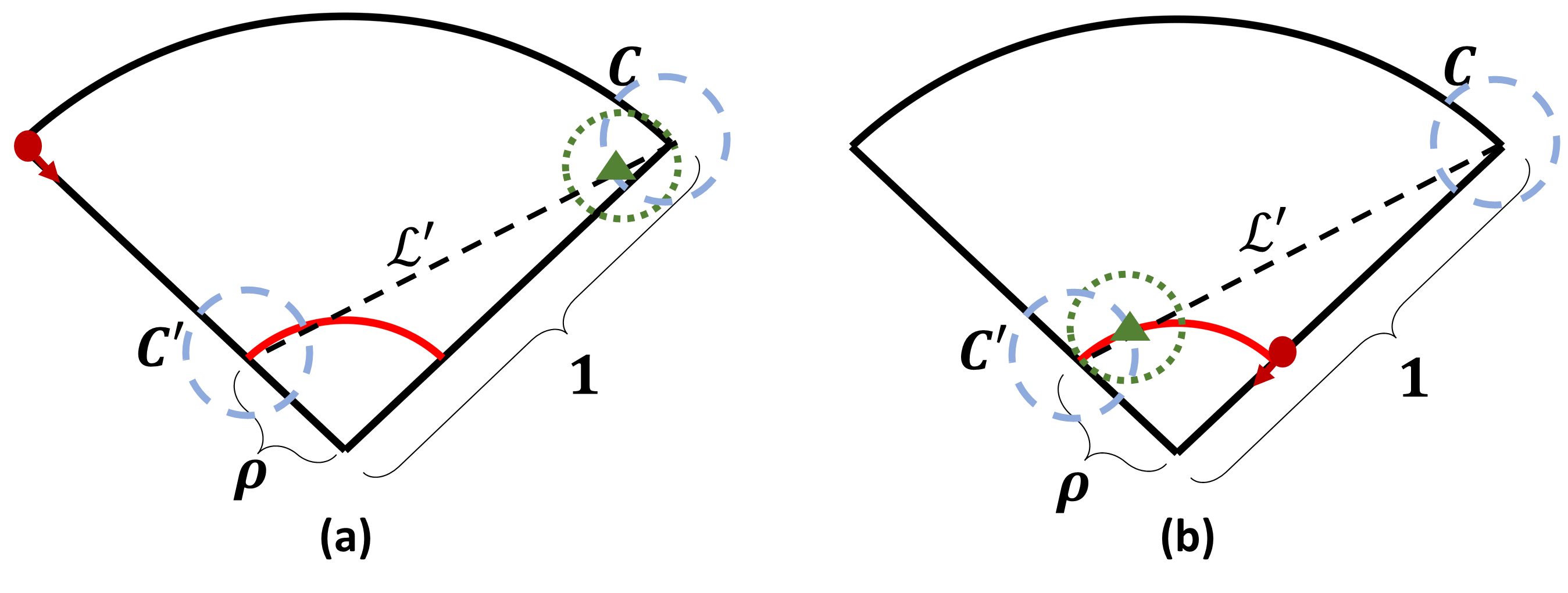}
    \caption{\small Description of the proof of Theorem \ref{thm:At_best_2} for $I_3$. The red curve denotes the perimeter. The circles $\mathcal{C}$ and $\mathcal{C'}$ are denoted by blue dashed circles and the line segment $\mathcal{L}'$ is denoted by the black dashed line. The vehicle and the intruders are denoted by a green triangle and a red dot, respectively. (a) The vehicle is located at $(t_1,\alpha_1)$ at time $t_1$. Intruder $b$ is at $(1,-\theta)$. (b) The vehicle is located at $(t_2,\alpha_2)$. Intruder b is captured but intruder a is lost.}
    \label{fig:At_best_2}
\end{figure}

We now establish a necessary condition for the existence of online algorithms having a competitive ratio of at least 2.
We first characterize locations $(t_1,\alpha_1)\in\mathcal{E}(\theta)$ and $(t_2,\alpha_2)\in\mathcal{E}(\theta)$ for the vehicle (Fig. \ref{fig:At_best_2}), where
\begin{align*}
    &t_1 = \sqrt{1+r^2-\tfrac{2r(1-\rho \cos(2\theta))}{\sqrt{1+\rho^2-2\rho \cos(2\theta)}}},\\
    & \alpha_1=\tan^{-1}\begin{pmatrix}\tfrac{\sin(\theta)\sqrt{1+\rho^2-2\rho \cos(2\theta)}-r(1+\rho)\sin(\theta)}{\cos(\theta)\sqrt{1+\rho^2-2\rho \cos(2\theta)}-r(1-\rho)\cos(\theta)}\end{pmatrix},\\
    &t_2=\sqrt{ \rho^2+r^2+\tfrac{2r\rho(\cos(2\theta)-\rho)}{\sqrt{1+\rho^2-2\rho \cos(2\theta)}}},\\
    &\alpha_2 = \tan^{-1}\begin{pmatrix}\tfrac{-\rho\sin(\theta)\sqrt{1+\rho^2-2\rho \cos(2\theta)}+r(1+\rho)\sin(\theta)}{\rho\cos(\theta)\sqrt{1+\rho^2-2\rho \cos(2\theta)}+r(1-\rho)\cos(\theta)}\end{pmatrix}.
\end{align*}
The locations $(t_1,\alpha_1)$ and $(t_2,\alpha_2)$ are determined analogously to the proof of Lemma \ref{lem:min_time}, and so, we only give an outline for it. Construct two circles $\mathcal{C}$ and $\mathcal{C}'$, each of radius $r$, centered at $(1,\theta)$ and $(\rho,-\theta)$ and consider a line segment $\mathcal{L'}$ that joins the centers of the two circles (Fig. \ref{fig:At_best_2}). Then, location $(t_1,\alpha_1)$ (resp. $(t_2,\alpha_2)$) corresponds to the intersection points of the line segment $\mathcal{L}'$ with circles $\mathcal{C}$ and $\mathcal{C}'$, respectively.
\begin{theorem}[Necessary condition for $c(\mathcal{P}) \geq 2$]\label{thm:At_best_2}
For any problem instance $\mathcal{P}(\theta,\rho,r,v)$,  $c(\mathcal{P})\geq 2$ if
\begin{align*}
\frac{1-\rho}{v}&\leq \sqrt{1+\rho^2-2\rho\cos(2\theta)}-2r,&\text{ if }& \theta\leq \frac{\pi}{2}\\
\frac{1-\rho}{v}&\leq 1+\rho-2r, &\text{ if }& \theta>\frac{\pi}{2}.
\end{align*} 
\end{theorem}
\begin{proof}
The key idea is to construct input instances for which any online algorithm is guaranteed to lose half the intruders while proving that an offline algorithm exists that can intercept all intruders. All of our input instances consists of two intruders denoted by $a$ and $b$ that arrive at location $(1,\theta)$ and $(1,-\theta)$, respectively, and we assume that the vehicle starts at the origin.

Two cases arise; (i) $\theta\leq \frac{\pi}{2}$ and (ii) $\theta>\frac{\pi}{2}$. We first consider case (i), i.e., $\theta\leq \frac{\pi}{2}$.

Consider that $\frac{1-\rho}{v} = \sqrt{1+\rho^2-2\rho\cos(2\theta)}-2r$ and consider an input instance $I_1$ in which both intruders $a$ and $b$ arrive at time instant $t_1$. This is the time that the vehicle takes to move from the origin directly to location $(t_1,\alpha_1)$.
We claim that the best way for any algorithm to capture both intruders is to capture either intruder $a$ or $b$ exactly at time $t_1$, i.e., as soon as it arrives and then move to capture the second intruder in minimum time. The explanation is as follows. 

The total time taken by the vehicle to capture both the intruders in the worst case is $\frac{1-x_i}{v}+\sqrt{x_i^2+\rho^2-2x_i\rho\cos(2\theta)}-2r$, where $\rho\leq x_i\leq 1$ is the radial component of the location of the first of the two intruders at the time of capture. The expression of the total time is determined as follows: The term $\frac{1-x_i}{v}$ is the intercept time for the first intruder. For the second term, 
construct a line segment $\mathcal{L}$ joining two points $(x_i,\theta)$ (resp. $(x_i,-\theta)$) and $(\rho,-\theta)$ (resp. $(\rho,\theta)$). Then, the length of the line segment $\mathcal{L}$ is given by $\sqrt{x_i^2+\rho^2-2x_i\rho\cos(2\theta)}$ from which we subtract $2r$ to account for the capture radius, to obtain the second term.
As $\frac{1-x_i}{v}+\sqrt{x_i^2+\rho^2-2x_i\rho\cos(2\theta)}-2r$ is monotonically decreasing function of $x_i$, its minimum is achieved at $x_i=1$. This establishes our claim that the minimum time any algorithm can take is to capture one intruder exactly when it arrives followed by the second intruder at $\sqrt{1+\rho^2-2\rho\cos(2\theta)}-2r$.

We now describe how an offline algorithm can capture both the intruders in the input instance $I_1$.
At time 0, the vehicle starts at the origin and moves towards location $(t_1,\alpha_1)$ capturing the intruder at location $(1,\theta)$ exactly at time $t_1$. Then the vehicle moves directly to location $(t_2,\alpha_2)$ exactly at time $t_1+\sqrt{1+\rho^2-2\rho\cos(2\theta)}-2r$ capturing the second intruder at $(\rho,-\theta)$. 
Note that placing the vehicle at $(t_1,\alpha_1)$ (resp. $(t_2,\alpha_2)$) ensures that the location $(1,\theta)$ (resp. $(\rho,-\theta)$) is on the circumference of the capture circle of the vehicle.
Therefore, an algorithm that hopes to be better than $2$-competitive must capture both the intruders in this input instance and the only way to do so is to move to either location $(t_1,\alpha_1)$ or $(t_1,-\alpha_1)$ arriving exactly at time $t_1$. 

Now consider input instances $I_2$ and $I_3$. In $I_2$, intruder $a$ arrives at time $t_1$ and intruder $b$ arrives at time $t_1+\epsilon$, where $\epsilon<L=2\sin(\theta)\begin{pmatrix}1-\frac{r(1+\rho)}{\sqrt{1+\rho^2-2\rho\cos(2\theta)}}\end{pmatrix}$ and $L$ denotes the minimum time required by the vehicle to move from $(t_2,\alpha_2)$ to $(t_2,-\alpha_2)$. In $I_3$, intruder $b$ arrives at time $t_1$ and intruder $a$ arrives at time $t_1+\epsilon$. Input instance $I_2$ (resp $I_3$) are constructed for algorithms that have the vehicle arriving at location $(t_1,-\alpha_1)$ (resp. $(t_1,\alpha_1)$) at time $t_1$.
Any algorithm that has the vehicle arriving at location $(t_1,-\alpha)$ (resp. $(t_1,\alpha_1)$) at time $t_1$ can capture only one intruder from $I_2$ (resp. $(I_3)$). 
As the solution is symmetric, we only provide the explanation for input instance $I_3$. This follows as the vehicle can capture intruder $b$ if it moves directly to location $(t_2,\alpha_2)$ (Fig. \ref{fig:At_best_2} (a)). However, as intruder $a$ arrives in at most $\epsilon<L$ time units, the vehicle will not be able to capture intruder $a$ (Fig. \ref{fig:At_best_2} (b)). An optimal offline algorithm can capture both the intruders by simply moving to $(t_1,-\alpha_1)$ at time $t_1$, capturing intruder $a$ upon arrival and then to $(t_2,-\alpha_2)$ to capture intruder $b$.

We now consider the case when $\frac{1-\rho}{v} < \sqrt{1+\rho^2-2\rho\cos(2\theta)}-2r$. Consider input instances $I_4$ and $I_5$. In $I_4$, intruder $a$ arrives at time $t_1$ and intruder $b$ arrives at time $t_1+\epsilon$, where $\epsilon=\sqrt{1+\rho^2-2\rho\cos(2\theta)}-2r-\frac{1-\rho}{v}$. In $I_5$, intruder $b$ arrives at time $t_1$ and intruder $a$ arrives at time $t_1+\epsilon$. Following similar reasoning as input instance $I_2$ and $I_3$, it follows that no online algorithm can capture both intruders from input instance $I_4$ or $I_5$.

We now consider case (ii), i.e., $\theta>\frac{\pi}{2}$. Except for when $\theta=\pi$, as the line segment $\mathcal{L}'$ will not be contained completely, the vehicle must move first to the origin and then to the next intercept point. Note that, the vehicle will do the same when $\theta=\pi$. Thus, in this case, the location $(t_1,\alpha_1)$ is $(1-r,\theta)$ and location $(t_2,\alpha_2)$ is $(1+\rho-2r,-\theta)$. Following similar steps as case (i), we construct input instances $I_1,\dots,I_5$ (omitted for brevity) and show that no online algorithm can capture both the intruders from those input instances.

In summary, even restricting our input instance to $\{I_1,\dots,I_5\}$, no online algorithm can capture both intruders whereas an optimal offline algorithm can capture both the intruders. This concludes the proof.
\end{proof}
\medskip
We now turn our attention to design of algorithms that provide sufficient conditions on the competitive ratios. In the next section, we design and analyze three algorithms, characterizing their parameter regimes with provably finite competitive ratios.

\section{Algorithms}\label{sec:Algorithms}
We start by defining an \emph{angular path} for the vehicle. Let the vehicle be located at $(x,\alpha)\in\mathcal{E}(\theta)$ for any $0<x\leq 1$ and $\alpha\in[-\theta,\theta]$. An angular path is a circular arc centered at the origin defined as $\mathcal{T}(x,\underline{\beta},\overline{\beta}):= \{ (x,\beta): \underline{\beta} \leq \beta \leq \overline{\beta} \}$ for any $\underline \beta, \overline \beta \in [-\theta,\theta]$ such that $\underline \beta \leq \alpha \leq \overline \beta$ and $\underline{\beta}\neq\overline{\beta}$. We say that the vehicle completes its motion on the angular path when the vehicle returns to its starting location after moving along all of the points in $\mathcal{T}$ twice.
Once to move from the starting location $(x,\alpha)$ to $(x,\overline \beta)$ (resp. $(x,\underline \beta)$), and second, to move from location $(x,\overline \beta)$ (resp. $(x,\underline \beta)$) to location $(x,\underline \beta)$ (resp. $(x,\overline \beta)$) and then back to the starting location $(x,\alpha)$.

\subsection{Angular Sweep algorithm}
Angular Sweep is an open loop algorithm, described as follows. The vehicle starts at location $(x_S,0)$, where 
\[
x_S \in \Big [ \frac{\rho-r}{1-a\theta v}, \min\{ 1-r,\rho+r \} \Big ],
\]
and $a=2$ if $\theta=\pi$ and $a=4$ if $\theta\neq \pi$.
This choice for the location $x_S$ will be justified shortly (Theorem~\ref{thm:sweep}). 

In Angular Sweep, the vehicle moves on an angular path with $x = x_S, \underline \beta = -\theta$ and $\overline \beta = \theta$ for any $\theta\neq \pi$. For $\theta=\pi$, the vehicle moves on a circle with $x_S$ as the radius and the origin as the center.

We first define the angular sweep algorithm for $\theta\neq \pi$. At time $0$, the vehicle first picks a velocity with unit magnitude and direction tangent to the angular path, oriented to the right until it reaches $(x_S,\theta)$. Once it reaches the endpoint, the vehicle switches direction and moves towards the other endpoint, $(x_S,-\theta)$. From this moment on, the vehicle only switches direction after it reaches an endpoint. In other words, the vehicle moves on the angular path $\mathcal{T}(x_s,-\theta,\theta)$, moving towards $(x_S,\theta)$ at time $0$.

We now define the algorithm for $\theta=\pi$. At time $0$, the vehicle picks a velocity with unit magnitude and direction perpendicular to its position vector, oriented to the right. From this point on, the vehicle keeps on moving in the direction perpendicular to its position vector for the entire duration, i.e., the vehicle moves on a  circle of radius $x_S$ and center as the origin.

\begin{theorem}[Angular Sweep competitiveness]\label{thm:sweep}
For any problem instance $\mathcal{P}(\theta,\rho,r,v)$ such that 
\begin{equation}\label{eq:v_sweep}
v \leq \min \Big \{ \frac{2r}{(\rho+r)a\theta}, \frac{1-\rho}{(1-r)a\theta} \Big \},
\end{equation}
where $a = 2$ (if $\theta = \pi$) or $a = 4$ (if $\theta \neq \pi$), with the choice of any
\[
x_S \in \Big [\frac{\rho-r}{1-a\theta v}, \min\{1-r,\rho+r\} \Big],
\]
Angular Sweep is $1$-competitive. Otherwise, Angular Sweep is not $c$-competitive for any constant $c$.
\end{theorem}
\begin{proof}
First, observe that if equation~\eqref{eq:v_sweep} holds, then the interval $ \Big [ \frac{\rho-r}{1-a\theta v}, \min\{1-r,\rho+r\} \Big]$ is non-empty and well defined. Therefore, it suffices to show that any $x_S$ from the said interval guarantees that Angular Sweep intercepts \emph{every intruder}. 

To justify the choice of $x_S$, we observe that there is no benefit for the vehicle to be located beyond a distance of $\rho+r$ and below $\rho-r$ from the origin. This follows because if the vehicle is located below $\rho-r$, then the capture circle will be completely below the perimeter and the vehicle cannot capture any intruder using the angular sweep algorithm. Moreover, for any radial location $x_S>\rho+r$, the vehicle will take $4\theta x_S$ time units to complete one angular path, whereas, in the worst case, the intruders will require $\frac{2r}{v}$ time to not get captured by the vehicle. Note that the time taken by the vehicle, i.e., $4\theta x_S$ increases as $x_S$ increases whereas the time taken by the intruders, i.e., $\tfrac{2r}{v}$ remains the same for any $x_S>\rho+r$. Thus, there will be no benefit for the vehicle to be located at a distance beyond $\rho+r$. 
To establish $1$-competitiveness, it is required that no intruders are lost by the vehicle and so, the total time taken by the vehicle to return to its starting location after completing its motion on the angular path must be at most the time required by the intruders to travel the distance of $x_S+r-\rho$. Mathematically, we require $a\theta x_S\leq \frac{x_S+r-\rho}{v}$, which implies $x_S\geq \frac{\rho-r}{1-a\theta v}$. Note that, since $\rho>r$, we require that $1-a \theta v>0$ or equivalently, $v<\frac{1}{a\theta}$. Finally, to ensure that $x_S+r$ is contained in the environment, we have $x_S\leq 1-r$. As, $\rho-r<\frac{\rho-r}{1-a\theta v}$, $\rho-r<\rho+r$, and $\rho-r<1-r$, the condition $x_S>\rho-r$ is always satisfied. Furthermore, since $\frac{1}{a\theta}<\frac{2r}{(\rho+r)a\theta}$ and $\frac{1}{a\theta}<\frac{1-\rho}{(1-r)a\theta}$, the condition $1-a\theta v>0$ always holds. Lastly, $x_S$ only exists if $\frac{\rho-r}{1-a\theta v}\leq\min\{ 1-r,\rho+r \}$ which yields $v\leq\min \{ \frac{2r}{(\rho+r)a\theta},\frac{1-\rho}{(1-r)a\theta} \}$.

We now prove that for any choice of $x_S\in [\frac{\rho-r}{1-a\theta v}, \min\{1-r,\rho+r\}]$, angular sweep algorithm is $1$-competitive.

Without loss of generality, we assume that, in the worst-case, at time instant $t$, the vehicle has just left the location $(x_S,\theta)$ and intruder $i$ is located at $(x_S+r,\theta)$.
The vehicle takes a total of $a\theta x_S$ time units to return to the location $(x_S,\theta)$ whereas the intruder takes $\frac{x_S+r-\rho}{v}$ time units to reach the perimeter. Thus, in order to ensure that the intruder $i$ is captured and takes time no less than $\frac{x_S+r-\rho}{v}$, we require $a\theta x_S\leq (x_S+r-\rho)/v$ and $x_S\leq 1-r$, respectively, which holds given that $x_S \in [\frac{\rho-r}{1-a\theta v}, \min\{1-r,\rho+r\}]$. 

For any $x_S\notin[\frac{\rho-r}{1-a\theta v}, \min\{1-r,\rho+r\}]$, we can construct an input instance with stream of intruders always arriving at $(1,\theta)$ such that when the vehicle leaves location $(x_S,\theta)$, an intruder is located at $(x_S+r,\theta)$. Since $x_S\notin[\frac{\rho-r}{1-a\theta v}, \min\{1-r,\rho+r\}]$, all intruders will be lost and the result follows. 
\end{proof}

\begin{figure}[t]
    \centering
    \includegraphics[scale=0.45]{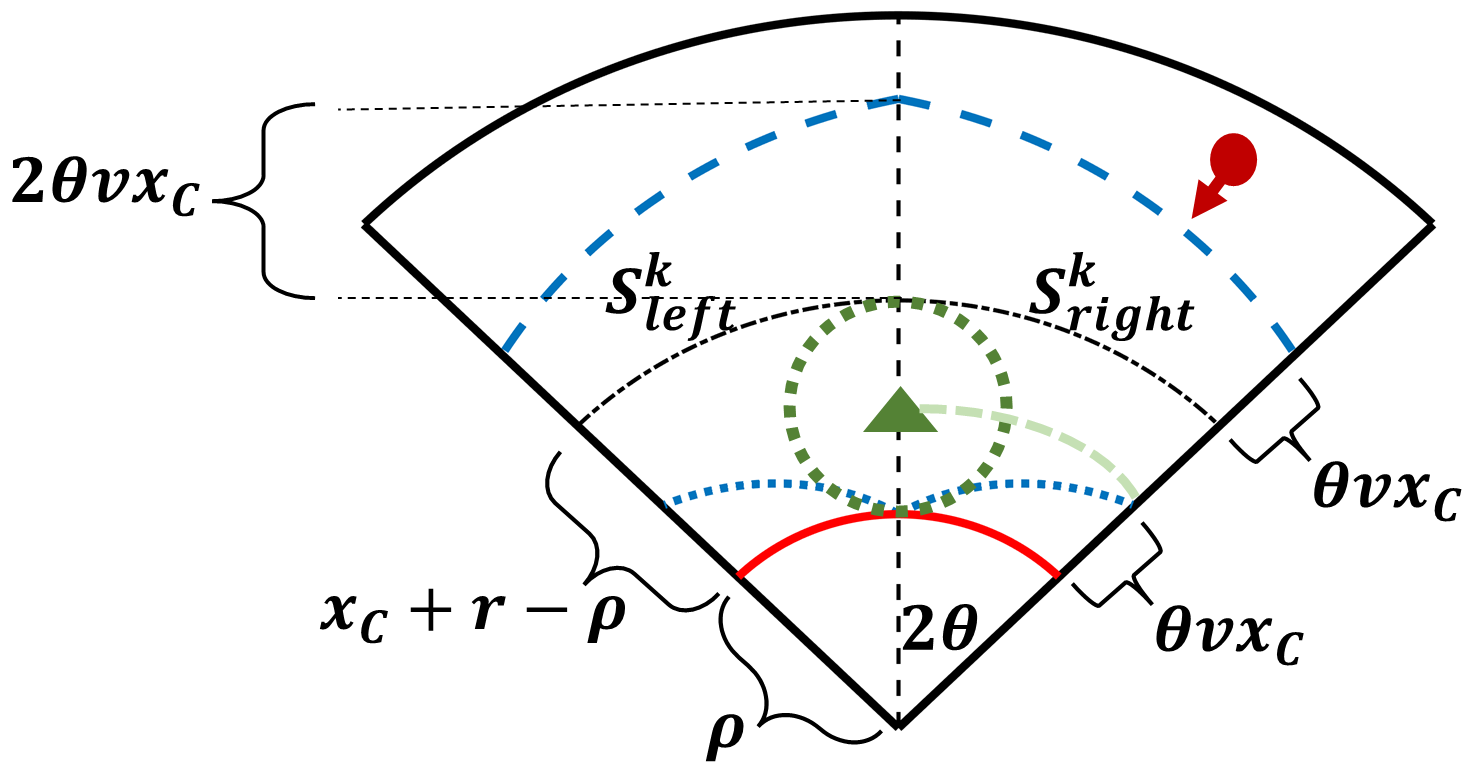}
    \caption{\small Setup for ConCaC algorithm for $x_C=r+\rho$. All intruders that are on the right side of the black dashed line and between the blue curves are in the set $\subscr{S}{right}^k$. All intruders that are on the left side of the black dashed line and between the blue curves are in the set $\subscr{S}{left}^k$. Green dashed curve denotes the angular path.}
    \label{fig:CaC_setup}
\end{figure}

\subsection{Conical Compare and Capture}
We now describe our second algorithm, Conical Compare and Capture (ConCaC) and establish that ConCaC is $2$-competitive for parameter regimes outside of those required for Angular Sweep to be $1$-competitive.

An epoch $k$ is defined as the time interval in which the vehicle completes its motion on angular path with a specified distance $x_C\in [\frac{\rho-r}{1-2\theta v}, \min\{ \rho+r,\frac{1-r}{1+v\theta} \} ]$ which is fixed for all epochs. The choice of $x_C$ will be justified shortly (Theorem \ref{thm:CaC}). ConCaC sets the parameters $\underline{\beta}$ and $\overline{\beta}$ for the angular path at the start of every epoch. Denote $|\subscr{S}{right}^k|$ (resp. $|\subscr{S}{left}^k|$) as the total number of intruders in the set $\subscr{S}{right}^k$ (resp. $\subscr{S}{left}^k$) in epoch $k$, where
\begin{multline*}
    \subscr{S}{right}^k(\rho,v) := \{( y,\beta):\rho+\beta x_{C}v< y\leq\\
    \min\{1,x_c+r+(2\theta-\beta)vx_C\}\forall \beta \in [0,\theta]\} \text{ and }\\
    \subscr{S}{left}^k(\rho,v) := \{( y,\beta):\rho-\beta x_{C}v < y \leq  \min\{1,x_c+r\\
    +(2\theta+\beta)vx_C\}\forall \beta \in (0,-\theta]\}.
\end{multline*}

\begin{algorithm}[t]
\DontPrintSemicolon
	\SetAlgoLined
	Select $x_C \in [ \frac{\rho-r}{1-2\theta v}, \min\{ \rho+r,\frac{1-r}{1+v\theta}] \}$.\\
	Wait until time $1-\min\{ 1,x_C+r+2\theta vx_C \}$.\\
	\For{each epoch $k\geq1$}{
	\eIf{$|\subscr{S}{left}^k|<|\subscr{S}{right}^k|$}{
	    Set $\underline{\beta}=0,\overline{\beta}=\theta$\\
		Move on angular path to location $(x_C,\theta)$\\
	Move on angular path to return to $(x_{C},0)$
		\;}
	{ 
	Set $\underline{\beta}=-\theta,\overline{\beta}=0$\\
	Move on angular path to location $(x_C,-\theta)$\\
	Move on angular path to return to $(x_{C},0)$
	\;
	}}
	\caption{Conical Compare-and-Capture Algorithm}
	\label{algo:CaC}
\end{algorithm}

Conical Compare and Capture algorithm is defined in Algorithm \ref{algo:CaC} and is summarized as follows.
At the start of every epoch $k$, the vehicle compares the total number of intruders in the set $\subscr{S}{left}^k$ and $\subscr{S}{right}^k$. If $|\subscr{S}{left}^k|<|\subscr{S}{right}^k|$, then the vehicle moves on the angular path ($x=x_C,\underline{\beta}=0, \overline{\beta}=\theta$) until it reaches location $(x_C,\theta)$ and then returns back to $(x_C,0)$, moving on the same angular path, capturing intruders on its way. Otherwise, the vehicle moves on an angular path ($x=x_C,\underline{\beta}=-\theta, \overline{\beta}=0$) towards the location $(x_C,-\theta)$. The vehicle then returns back to $(x_C,0)$ moving on the same angular path. After returning to $(x_C,0)$, the vehicle repeats the same for the next epoch.

For the initial case, we assume time $0$ as the time when the first intruder arrives in the environment. The vehicle starts at location $(x_C,0)$ and waits for $1-\min\{1,x_C+r+2\theta vx_C \}$ amount of time and then begins its first epoch.

\begin{lemma}\label{lem:CaC_finite_compt_condtn}
Any intruder that lies beyond\footnote{ intruders with radial coordinate more than $x_c+r+(2\theta-\beta)v x_C$} the location $(x_c+r+(2\theta-\beta)v x_C,\beta),~\forall \beta \in [-\theta,\theta]$ in epoch $k$, will either be contained in the set $\subscr{S}{left}^{k+1}$ or in $\subscr{S}{right}^{k+1}$ in epoch $k+1$ and is not lost at the start of epoch $k+1$ if $v\leq \frac{x_C+r-\rho}{2\theta x_C}$.
\end{lemma}
\begin{proof}
Without loss of generality, assume that $\abs{\subscr{S}{left}^k}<\abs{\subscr{S}{right}^k}$ at epoch $k$. The total time taken by the vehicle to capture intruders in $\subscr{S}{right}^k$ and return back to its starting location $(x_C,0)$ is $2\theta x_C$. In the worst-case, in order for any intruder $i$ to be not considered in the start of epoch $k$, the intruder $i$ must be located just above $(x_C+r+\theta vx_C,\theta)$ at the start of epoch $k$, i.e., $\beta=\theta$. By the time the vehicle reaches location $(x_C,\theta)$, intruder $i$ will be located just above the location $(x_C+r,\theta)$ and will not be captured. Since $v\leq \frac{x_C+r-\rho}{2\theta x_C}$, the intruder will be at least $\theta v x_C$ distance away from the perimeter at the end of epoch $k$. Thus, this intruder will be considered in the start of epoch $k+1$ given the definition of set $\subscr{S}{left}^{k+1}$ and $\subscr{S}{right}^{k+1}$. Clearly, as the intruder $i$ will be considered for comparison in epoch $k+1$, it is not lost unless the vehicle decides to move to $\subscr{S}{left}^{k+1}$ in epoch $k+1$. This concludes the proof.
\end{proof}

\begin{theorem}[ConCaC competitiveness] \label{thm:CaC}
For any problem instance $\mathcal{P}(\theta,\rho,r,v)$ such that 
\begin{equation}\label{eq:v_CaC}
    v\leq \min \Big\{ \frac{r}{\theta (\rho+r)},\frac{1-\rho}{\theta(2-3r+\rho)} \Big \}, 
\end{equation}
with the choice of any 
\[x_C\in \Big [\frac{\rho-r}{1-2\theta v}, \min\{ \rho+r,\frac{1-r}{1+v\theta} \}\Big ],\]
ConCaC algorithm is $2$-competitive. 
\end{theorem}
\begin{proof}
First, observe that if equation~\eqref{eq:v_CaC} holds, then the interval $[\frac{\rho-r}{1-2\theta v}, \min\{\frac{1-r}{1+v\theta},\rho+r\}]$ is non-empty. Therefore, it suffices to show that for any $x_C$ from the said interval, ConCaC algorithm is $2$-competitive. 

We first justify the choice of $x_C$. Similar to the proof of Theorem \ref{thm:sweep}, we observe that there is no benefit for the vehicle to be located beyond a distance of $\rho+r$ and below $\rho-r$ from the origin. 
From Lemma \ref{lem:CaC_finite_compt_condtn}, in order to ensure that every intruder that is not considered in the set $\subscr{S}{left}^k$ and $\subscr{S}{right}^k$ in epoch $k$, is considered in either set $\subscr{S}{left}^{k+1}$ or $\subscr{S}{right}^{k+1}$ in epoch $k+1$, we require $x_C\geq \frac{\rho-r}{1-2\theta v}$.
Note that, since $\rho>r$, we require that $1-2 \theta v>0$ or equivalently, $v<\frac{1}{2\theta}$. Finally, to ensure that $x_C+r+x_C v\theta$ is contained in the environment, we require $x_C\leq \frac{1-r}{1+v\theta}$. 
Note that $\rho-r<\frac{\rho-r}{1-2\theta v}$, $\rho-r<\rho+r$. Also, if equation \eqref{eq:v_CaC} holds then, $\rho-r<\frac{1-r}{1+v\theta}$. Thus, the condition $x_C>\rho-r$ always holds. Furthermore, since $\frac{1}{2\theta}<\frac{r}{(\rho+r)\theta}$ and $\frac{1}{2\theta}<\frac{1-\rho}{\theta(2-3r+\rho)}$, the condition $1-2\theta v>0$ always holds. Lastly, $x_C$ only exists if $\frac{\rho-r}{1-2\theta v}\leq\min\{ \frac{1-r}{1+v\theta},\rho+r \}$ which yields $v\leq\min \{ \frac{r}{(\rho+r)\theta},\frac{1-\rho}{\theta(2-3r+\rho)} \}$.

We now prove that for any choice of $x_C\in [\frac{\rho-r}{1-2\theta v}, \min\{\frac{1-r}{1+v\theta},\rho+r\}]$, ConCaC algorithm is $2$-competitive.

Lemma \ref{lem:CaC_finite_compt_condtn} ensures that every intruder will belong to either set $\subscr{S}{left}^k$ or $\subscr{S}{right}^k$ in every epoch $k$. 
In every epoch $k$, the vehicle compares the total number of intruders on either side contained in the set $\subscr{S}{left}^k$ and $\subscr{S}{right}^k$ and moves to the side where the number of intruders is higher. Thus, it is guaranteed that the vehicle will capture at least half of the total number of intruders that arrive in the environment, assuming that an optimal offline algorithm can capture all intruders.
\end{proof}


\begin{figure}[t]
    \centering
    \includegraphics[scale=0.29]{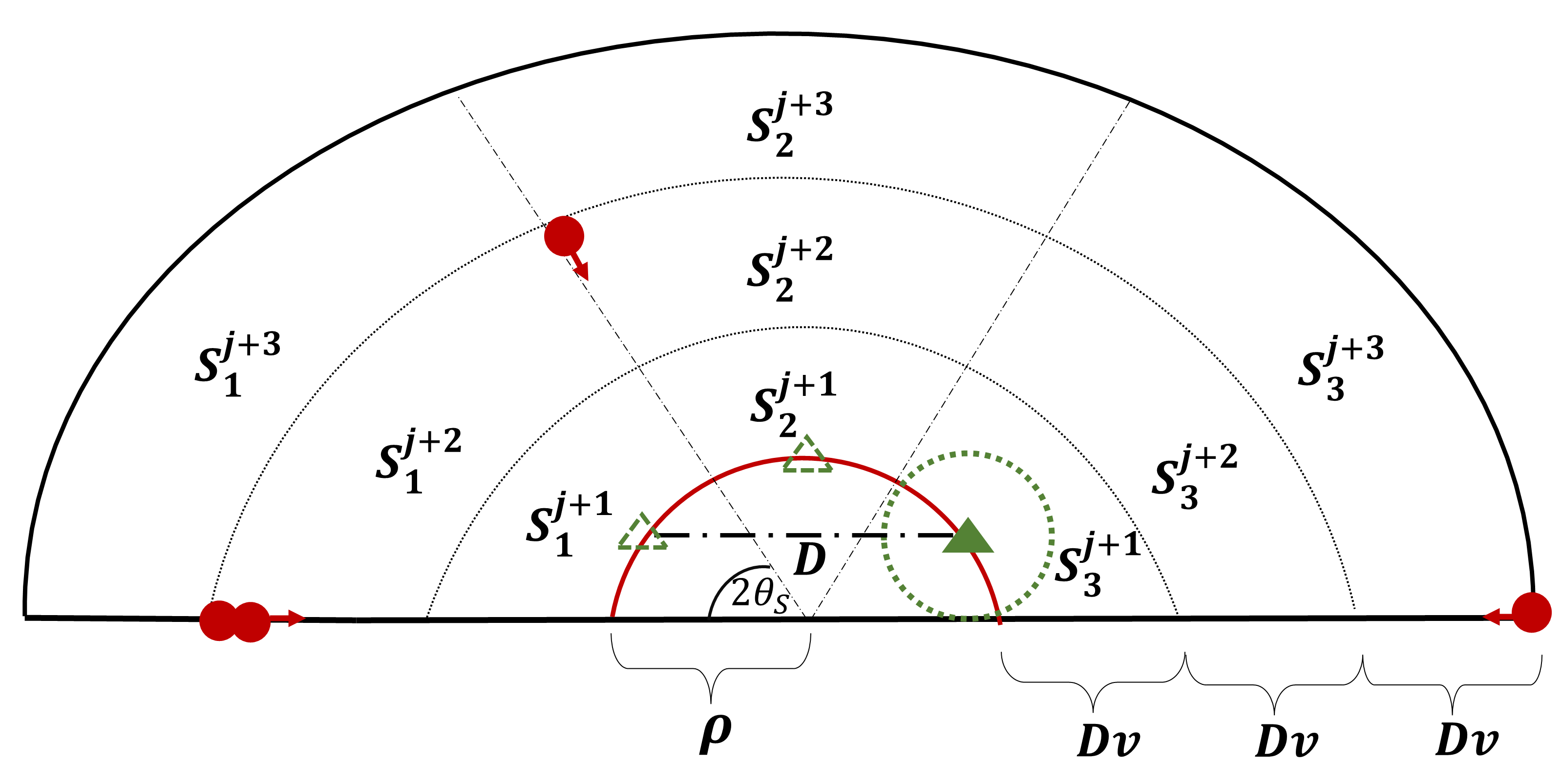}
    \caption{\small Breakdown of environment into $n_s=3$ sectors and time intervals of length $D$. The dashed green triangles denote the resting point of each sector. Vehicle is located at the resting point $(x_3,\alpha_3)$ of sector $N_3$.}
    \label{fig:SNP}
\end{figure}
\subsection{Stay Near Perimeter}

Unlike the previous two algorithms, in this algorithm, the vehicle does not follow an angular path. Instead, the idea is to divide the environment into sectors and position the vehicle close to the perimeter in a specific sector.

We partition the environment $\mathcal{E}(\theta)$ into $n_s=\lceil\frac{\theta}{\theta_s}\rceil$ sectors, each with angle $2\theta_s=2\arctan(\frac{r}{\rho})$. Since $r<\rho$, $\theta_s<\frac{\pi}{4}$. Let $N_l, l\in \{1,\dots,n_s\}$ denote the $l^{th}$ sector, where $N_1$ corresponds to the leftmost sector in the environment (Fig. \ref{fig:SNP}). Then, a \textit{resting point} $(x_l,\alpha_l)\in \mathcal{E}(\theta)$ of a sector $N_l$ is defined as the location for the vehicle such that when positioned at that location, the entire perimeter contained in that sector is contained completely within the capture radius of the vehicle. Mathematically, the resting point, $(x_l,\alpha_l)$, for a sector $N_l$ is defined as $(\tfrac{\rho}{\cos(\theta)},(l-\tfrac{n_s+1}{2})2\theta_s)$.
Further, we define $D$ as the distance between the two resting points that are farthest in the environment as
\begin{equation}
    D =
    \begin{cases}
    2\frac{\rho}{\cos(\theta_s)}\sin((n_s-1)\theta_s), \text{ if } (n_s-1)\theta_s < \frac{\pi}{2}\\
    2\frac{\rho}{\cos(\theta_s)}, \text{ otherwise.}
    \end{cases}\label{def:Dist}
\end{equation}
Note that $n_s=1\Rightarrow D=0$. This means that there is only one sector, i.e., the environment and the capture circle can contain the entire perimeter. Thus, the vehicle is to be positioned at the unique corresponding resting point and must capture all intruders that  arrive in the environment.

\begin{algorithm}[t]
\DontPrintSemicolon
	\SetAlgoLined
	Stay at origin until time $D$. \\
	$k^* = \argmax_{k\in \{1,\dots,n_s \}} \{\eta_i^1,\dots,\eta_i^{n_s}\}$,
	$N_i=N_{k^*}$\\
	Move to $(x_i,\alpha_i)$\\
	Wait until time $3D$.\\
	Assumes vehicle is at $(x_i,\alpha_i)$ in sector $N_i$\\
	\For{each $j\geq1$}{
	$k^* = \argmax_{k\in \{1,\dots,n_s \}} \{\eta_i^1,\dots,\eta_i^{n_s}\}$\\
	$N_o=N_{k^*}$\\
    \eIf{$N_o\neq N_i$ and $\abs{S_{o}^{j+2}}\geq \abs{S_{i}^{j+1}}$}{
	    Move to $(x_o,\alpha_o)$\\
	    Capture $\abs{S_{o}^{j+2}}$
		\;}
	{ Stay at $(x_i,\alpha_i)$\\
	  Capture $\abs{S_{i}^{j+1}}$
	\;}
	}
	\caption{Stay Near Perimeter Algorithm}
	\label{algo:SNP}
\end{algorithm}

After partitioning the environment into $n_s$ sectors, Stay Near Perimeter (SNP) algorithm divides the environment into three time intervals of time length $D$ each. Specifically, the $j^{th}$ interval for any $j>0$ is defined as the time interval $[(j-1)D,jD]$. In order to ensure a finite competitiveness for this algorithm, we require $\frac{1-\rho}{v}\geq 3D$, i.e., the intruders require at least $3D$ time to reach the perimeter.
For any $j\geq1$, let $S_l^j$ be the set of intruders that arrive in a sector $N_l$ in the $j^{th}$ interval (Fig \ref{fig:SNP}).

The SNP algorithm (defined in Algorithm \ref{algo:SNP}) is based on the following two steps: First, select a sector in the environment with maximum number of intruders. Second, determine if it is beneficial to switch over to that sector. 
These two steps are achieved by two simple comparisons; \textbf{C1} and \textbf{C2} detailed below.

In the first comparison \textbf{C1}, SNP determines that sector which has the most number of intruders in the last two intervals as compared to the total number of intruders in the entire sector in which the vehicle is located. In particular, suppose that the vehicle is located at the resting point of sector $N_i$ at the $j$-th iteration. Corresponding to any sector $N_l$, we define $\eta_i^l$ as
\[
\eta_i^l\triangleq \begin{cases} \abs{S_l^{j+2}}+\abs{S_l^{j+3}}, &\text{if } l \neq i,\\
\abs{S_i^{j+1}}+\abs{S_i^{j+2}}+\abs{S_i^{j+3}}, &\text{if } l = i.
\end{cases}
\]
Then, SNP selects the sector $N_{k^*}$, where $k^* = \argmax_{k\in \{1,\dots,n_s \}} \{\eta_i^1,\dots,\eta_i^{n_s}\}$. 
In case there are multiple sectors with same number of intruders, then SNP breaks the tie as follows. If the tie includes the sector $N_i$, then SNP selects $N_i$. Otherwise, SNP selects the sector with the maximum number of intruders in the interval $j+2$. If this results in another tie, then this second tie can be resolved by picking the sector with the least index. Let the sector chosen as the outcome of \textbf{C1} be $N_o, o\in \{1,\dots,n_s \}$. 

For the second comparison \textbf{C2}, if the sector chosen is $N_o, o \neq i$, and the total number of intruders in the set $S_o^{j+2}$ is no less than the total number of intruders in $S_i^{j+1}$, then SNP moves the vehicle to $(x_o,\alpha_o)$ arriving in at most $D$ time units. Then the vehicle waits at that location to capture all intruders in $S_o^{j+2}$. Otherwise (i.e., if $S_o^{j+2} < S_i^{j+1}$ or $o = i$), the vehicle stays at its current location $(x_i,\alpha_i)$, captures intruders in $S_i^{j+1}$ and then reevaluates after time interval of $D$. 

At time $0$, the vehicle waits for $D$ time units at location $(0,0)$ after the first intruder arrives in the environment. Then the vehicle moves to the sector which has the maximum number of intruders in $S_i^1,~\forall N_i$ sectors in the environment. The vehicle then waits until time $3D$. To ensure that no intruder is lost until time $3D$, we require $\tfrac{\rho}{\cos(\theta_s)}\leq 2D$.

\begin{lemma}\label{lem:1captureset}
Let the vehicle be located at a resting point $(x_i,\alpha_i)$ of a sector $N_i, i\in \{ 1,\dots,n_s \}$. Then, for any $j\geq 1$, the vehicle always captures intruders in either $S_i^{j+1}$ or $S_o^{j+2}$, where $N_o$ denotes the sector selected by SNP after \textbf{C1}.
\end{lemma}
\begin{proof}
Consider that the sector $N_o=N_i$. Then, according to Algorithm \ref{algo:SNP}, the vehicle stays at its current position and captures $S_i^{j+1}$ and the result follows.

Now consider that the sector $N_o\neq N_i$. Then there are two cases: (i) Either the vehicle decides to stay at its current position for $D$ time interval, i.e., $\abs{S_i^{j+1}}>\abs{S_o^{j+2}}$ or (ii) the vehicle decides to move to the resting point corresponding to the sector $N_o$, i.e., $\abs{S_i^{j+1}}\leq \abs{S_o^{j+2}}$. In case (i), the vehicle stays at its current location and captures $\abs{S_i^{j+1}}$. In case (ii), the vehicle spends at most $D$ time units to moves to the resting point of the sector $N_o$ and then captures intruders in the set $S_o^{j+2}$. This concludes the proof.
\end{proof}

To establish the competitive ratio of Algorithm SNP, we use an accounting analysis in which captured intervals \textit{pay} for the lost intervals or equivalently, captured intervals are \textit{charged} for the intervals lost. 
The following lemmas will jointly establish the competitive ratio of SNP algorithm.

\begin{figure*}
\centering
\begin{subfigure}[Type (a) captured intervals. Vehicle stays at $(x_3,\alpha_3)$ to capture $S_3^{j+1}$ and $S_3^{j+2}$.]{\includegraphics[width=0.4\textwidth]{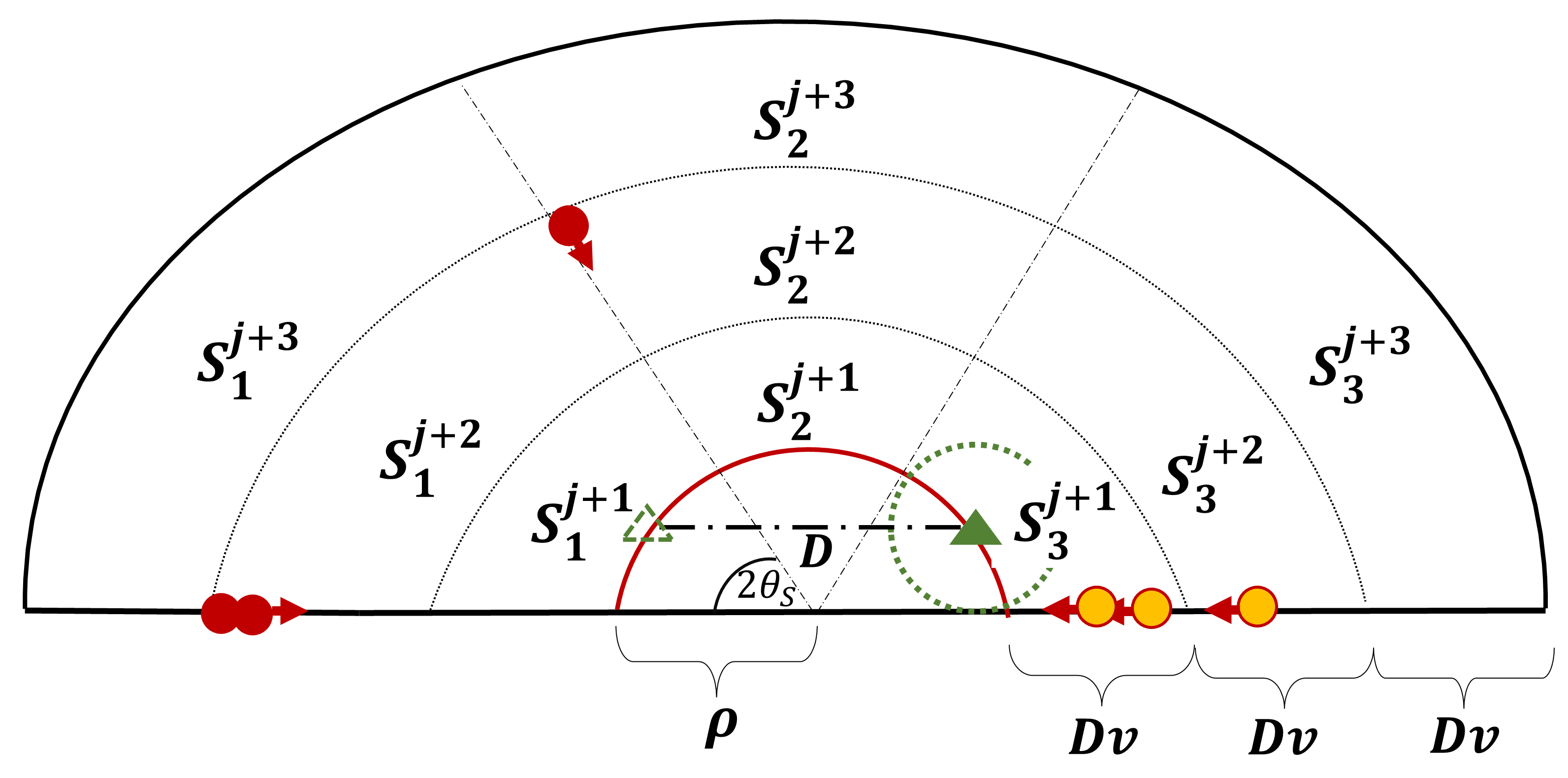}}
\end{subfigure}
\subfigure[Type (b) captured intervals. Vehicle stays at $(x_3,\alpha_3)$ to capture $S_3^{j+1}$ and then moves to $(x_1,\alpha_1)$ to capture $S_1^{j+3}$.]{\includegraphics[width=0.4\textwidth]{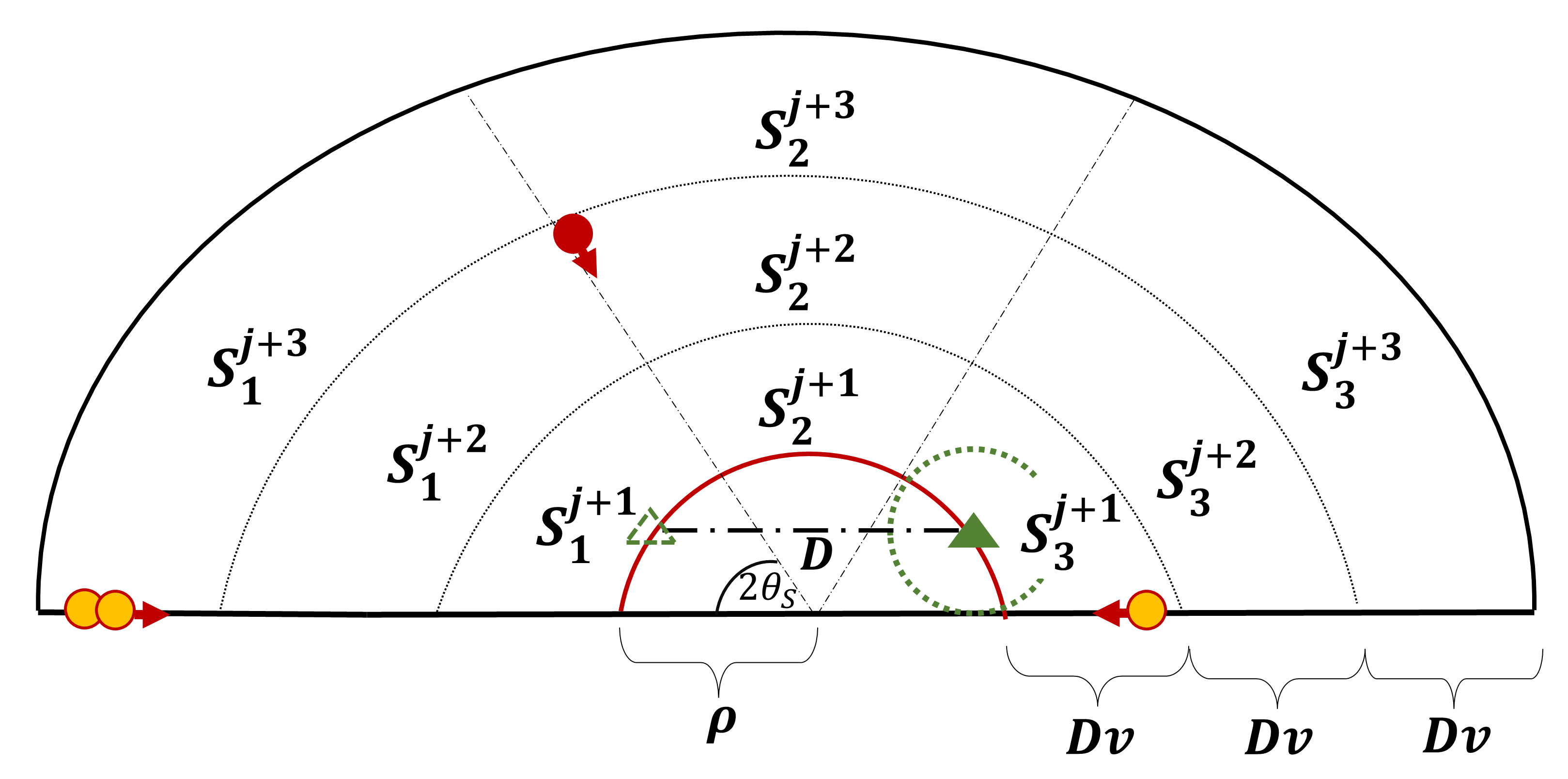}}
\subfigure[Type (c) captured intervals. Vehicle moves to $(x_1,\alpha_1)$ to capture $S_1^{j+2}$ and $S_1^{j+3}$.]{\includegraphics[width=0.4\textwidth]{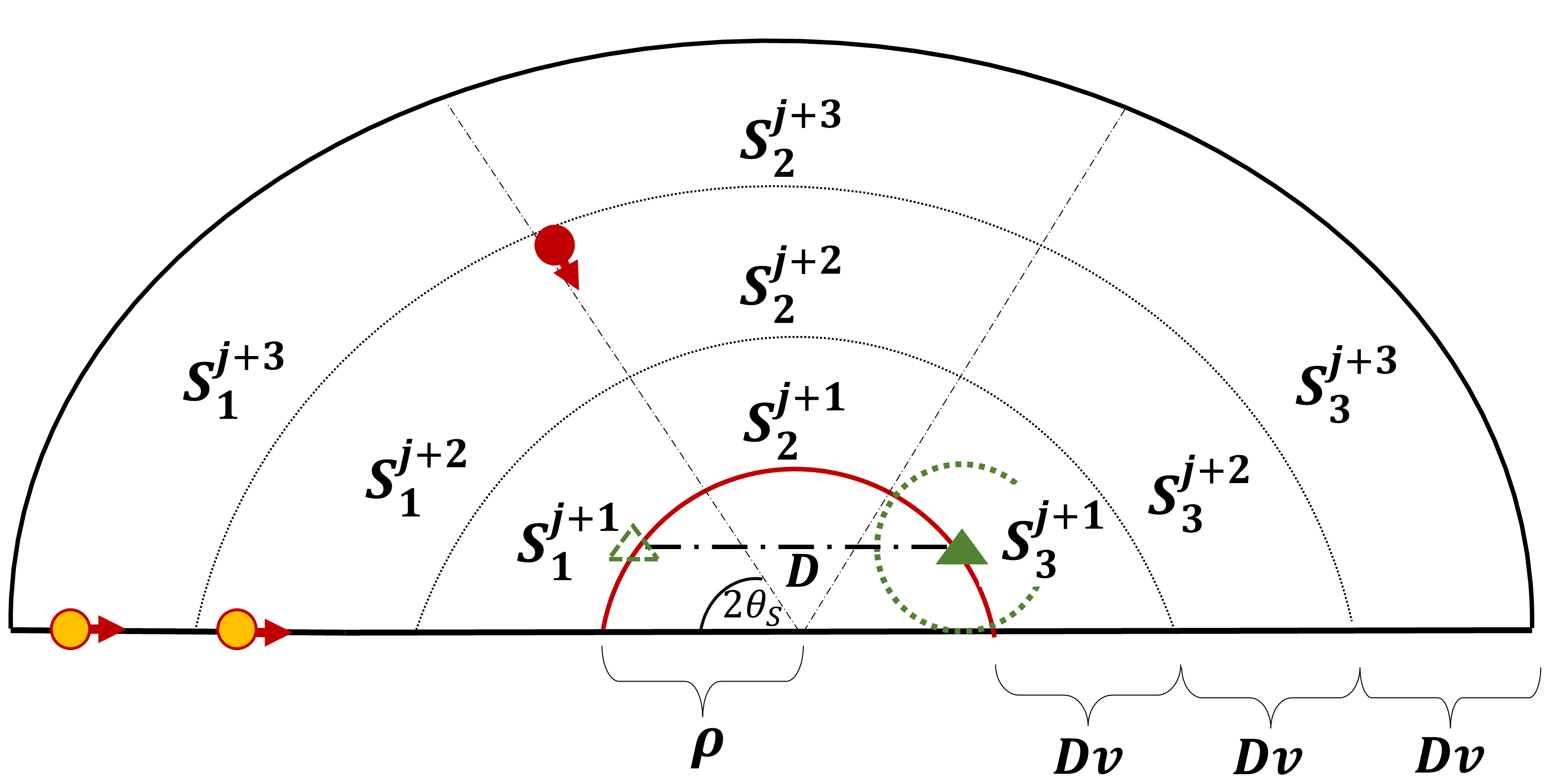}}
\subfigure[Type (d) captured intervals. Vehicle moves to $(x_1,\alpha_1)$ to capture $S_1^{j+2}$ and then to $(x_3,\alpha_3)$ to capture $S_3^{j+4}$.]{\includegraphics[width=0.4\textwidth]{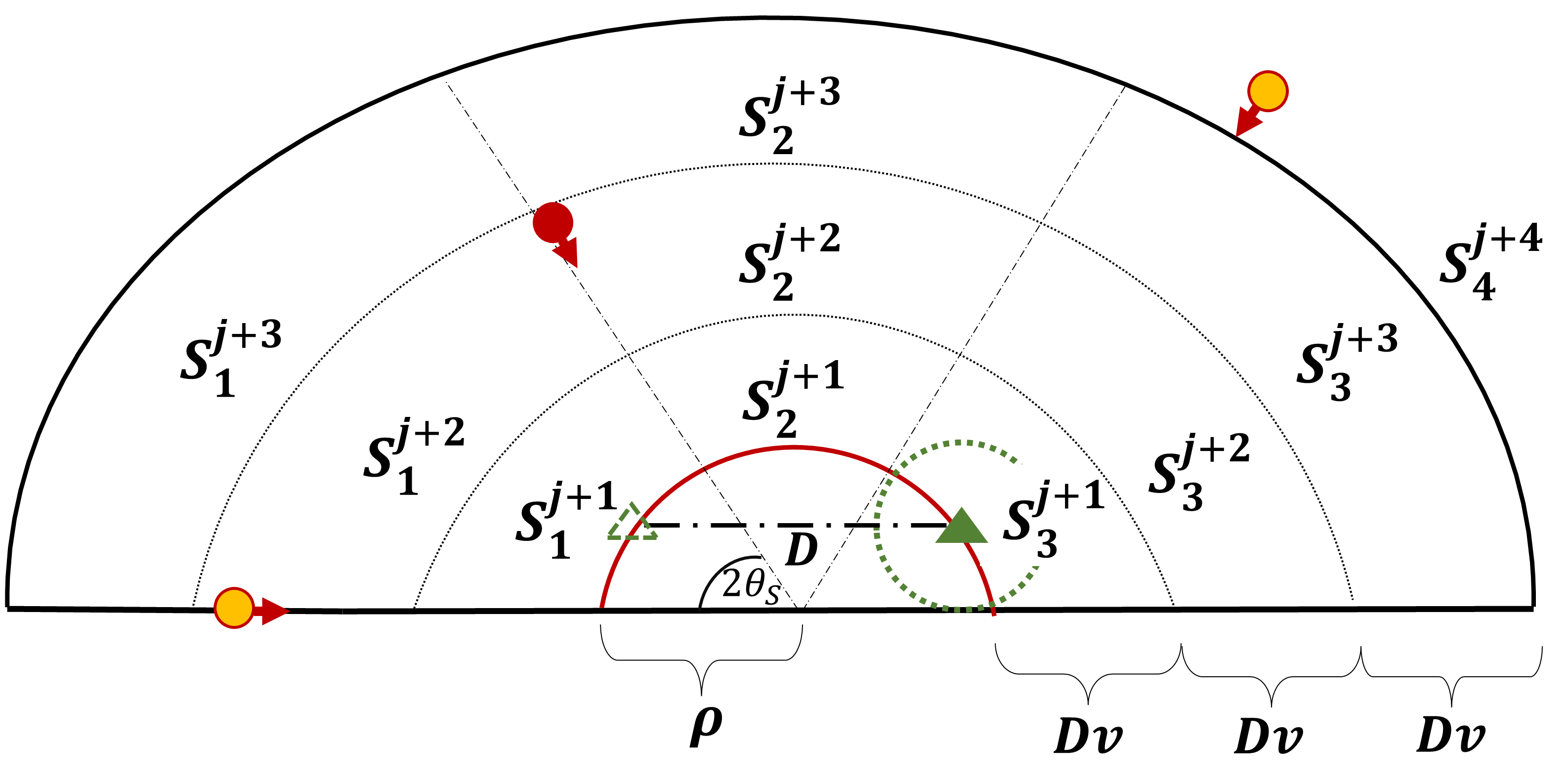}}
\caption{Depiction of captured intervals for the proof of Lemma \ref{lem:SNP_part1} for $i=3$ and $o=1$. The intruders of the respective captured intervals are represented by yellow circles. Note that for type (d) captured intervals,  intruders of interval $\subscr{S}{4}^{j+4}$ have not arrived in the environment. Information of interval $\subscr{S}{4}^{j+4}$ is revealed to the SNP algorithm once the vehicle moves to $(x_1,\alpha_1)$.}
\label{fig:SNP_capt_int}
\end{figure*}

\begin{lemma}\label{lem:SNP_part1}
In algorithm SNP, any two consecutive captured intervals pay for total $3(n_s-1)$ lost intervals.
\end{lemma}
\begin{proof}
As Lemma \ref{lem:1captureset} ensures that the vehicle always captures an interval of intruders, any two consecutive captured intervals can be classified into four types (Fig. \ref{fig:SNP_capt_int}); (a) stay at the current location and capture both intervals on the same side, (b) stay at the current location and capture an interval and then move to the resting point of $N_o$ and capture the second interval, (c) move to the resting point of $N_o$ and capture both intervals, and finally (d) move to the resting point of sector $N_o$ and capture an interval and then move to the resting point of another sector, $N_{o'},~o'\in \{1,\dots,n_s \}\setminus\{o\}$ and capture an interval.

The explanation for Type (a)  captured intervals $S_i^{j+1}$ and $S_i^{j+2}$ is as follows. At time instant $jD$ and $(j+1)D$, since vehicle decides to capture $S_i^{j+1}$ and $S_i^{j+2}$ (comparison $\textbf{C1}$ and $\textbf{C2}$), it loses $S_l^{j+2}$ and $S_l^{j+3}$ intruders from other sectors, i.e., $\forall l\in \{1,\dots,n_s\}\setminus\{i\}$. Thus the captured intervals $S_i^{j+1}$ and $S_i^{j+2}$ are charged $2n_s-2$ times. The remaining $n_s-1$ charge is explained as follows. Since the vehicle is currently located at $(x_i,\alpha_i)$ it must be that the vehicle captured $S_i^j$. This implies that comparison $\textbf{C1}$ must have yielded sector $N_i$ at either time instant $(j-2)D$ (if the vehicle was located at $(x_l,\alpha_l),l\neq i)$) or $(j-1)D$ (if the vehicle was located at $(x_i,\alpha_i)$). Recall that $\textbf{C1}$ requires at least $S_i^j$ and $S_i^{j+1}$ for the comparison. As the vehicle captured $S_i^j$, the captured interval $S_i^{j+1}$ is charged another $n_s-1$ times for both $S_l^{j}$ and $S_l^{j+1}$ combined for all $l\neq i$.

Following similar calculations, type (b) captured intervals $S_i^{j+1}$ and $S_o^{j+3}$ are also charged $3(n_s-1)$ times. $n_s-1$ times to pay for lost intervals $S_l^{j}$ and $S_l^{j+1}$ combined and $n_s-1$ times for lost interval $S_l^{j+2}$, $\forall l\in \{1,\dots,n_s\}\setminus\{i\}$. The remaining $n_s-1$ pay is as follows. Once for all lost intervals $S_i^{j+2}$, $S_i^{j+3}$, and $S_i^{j+4}$ combined and $n_s-2$ pay for lost intervals $S_{l'}^{j+3}$, and $S_{l'}^{j+4}$ combined $\forall l'\in \{1,\dots,n_s\}\setminus\{i,o\}$ (comparison \textbf{C1} and \textbf{C2} at time $(j+1)D$). 

Type (c) captured intervals $S_o^{j+2}$ and $S_o^{j+3}$ pay once for lost intervals $S_i^{j+1}$, $S_i^{j+2}$, and $S_i^{j+3}$ combined as well as $n_s-2$ times for the lost intervals $S_l^{j+2}$ and $S_l^{j+3}$, $\forall l \in \{1,\dots,n_s\}\setminus\{ i,o\}$ (comparison \textbf{C1} and \textbf{C2} at time $jD$). The captured intervals also pay $n_s-1$ times for lost intervals $S_l^{j+4}$ for all $N_l,l\neq o$ sectors. Finally, the last $n_s-1$ pay is for lost interval $S_{l'}^{j}$ and $S_{l'}^{j+1}$, $\forall l'\in \{1,\dots,n_s\}\setminus\{i\}$ as the vehicle captured $S_o^{j+2}$ instead of $S_i^{j+1}$ (comparison $\textbf{C1}$).

For type (d) captured intervals, without loss of generality, consider that after capturing its first interval, $S_o^{j+2}$, in sector $N_o$, the vehicle moves back to sector $N_i$ to capture its second interval $S_i^{j+4}$, i.e., $N_{o'}=N_i$. Type (d) captured interval $S_o^{j+2}$ pays once for $S_i^{j+1}$, $S_i^{j+2}$, and $S_i^{j+3}$ combined and $n_s-2$ times for the lost intervals $S_l^{j+2}$ and $S_l^{j+3}$ combined, $\forall l\in\{1,\dots,n_s \}\setminus\{i,o\}$ (comparison \textbf{C1} and \textbf{C2} at time $j$).
The captured interval $S_i^{j+4}$ pays once for $S_o^{j+3}$, $S_o^{j+4}$, and $S_o^{j+5}$ combined and $n_s-2$ times for the lost intervals $S_l^{j+4}$ and $S_l^{j+5}$ combined (comparison \textbf{C1} and \textbf{C2} at time $j+2$). 
The final pay is $n_s-1$ times for lost intervals $S_{l'}^{j}$ and $S_{l'}^{j+1}$ combined, $\forall l'\in\{1,\dots,n_s\}\setminus\{i\}$ as the vehicle captured $S_o^{j+2}$ and instead of $S_i^{j+1}$ (comparison $\textbf{C1}$). 

Since each type of captured intervals are charged $3(n_s-1)$ times, the result is established.
\end{proof}

We now establish that each lost interval is fully accounted for by the captured intervals.
Since SNP directs the vehicle to stay at a resting point of any sector for some time interval, it can be viewed as a sequence of \textit{traces}, in which the vehicle spends some number of intervals at one resting point and some number of intervals at another. Each trace is defined by a set $\{ k_1, k_2, \dots, k_{n_s} \}$, where each element $k_l$, $ l\in \{ 1,\dots, n_s \}$ denotes the number of intervals that the vehicle decides to capture by staying at the corresponding resting point of the sector $N_l$. 

\begin{figure*}[t]
\centering
\begin{subfigure}[$\theta=\pi/4$]{\includegraphics[width=0.3\textwidth]{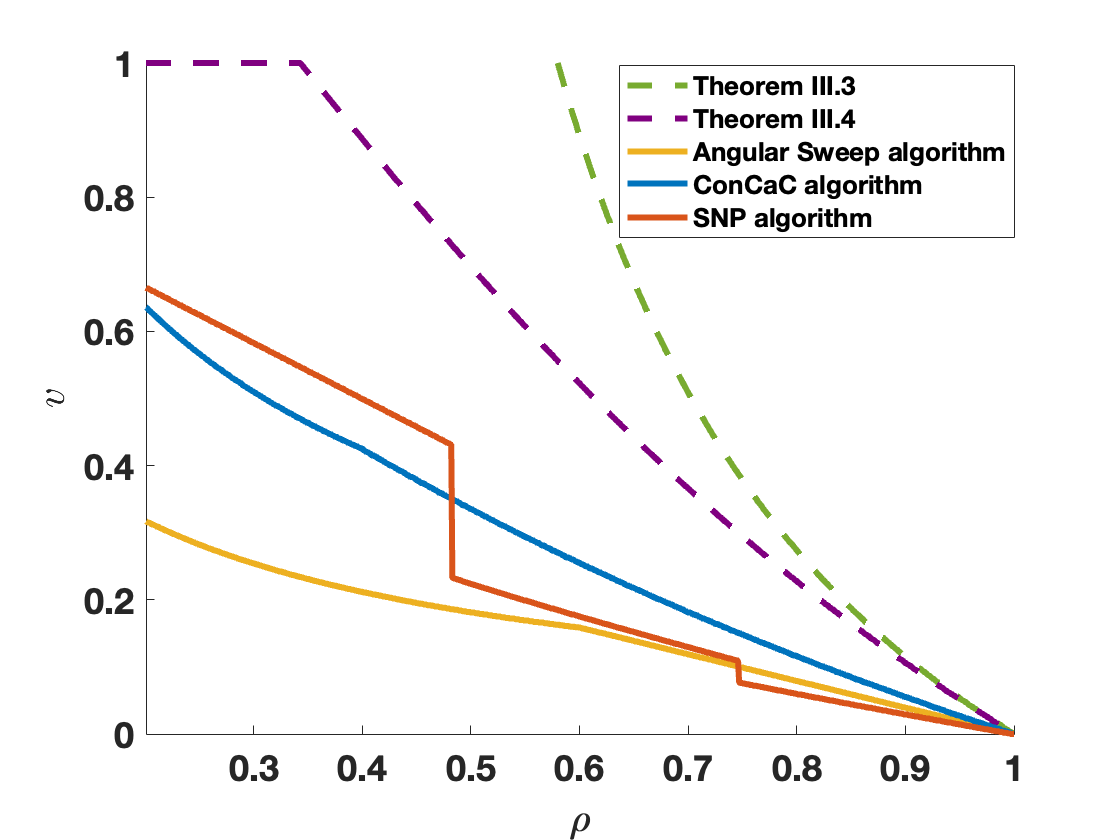}}
\end{subfigure}
\subfigure[$\theta=\pi/3$]{\includegraphics[width=0.3\textwidth]{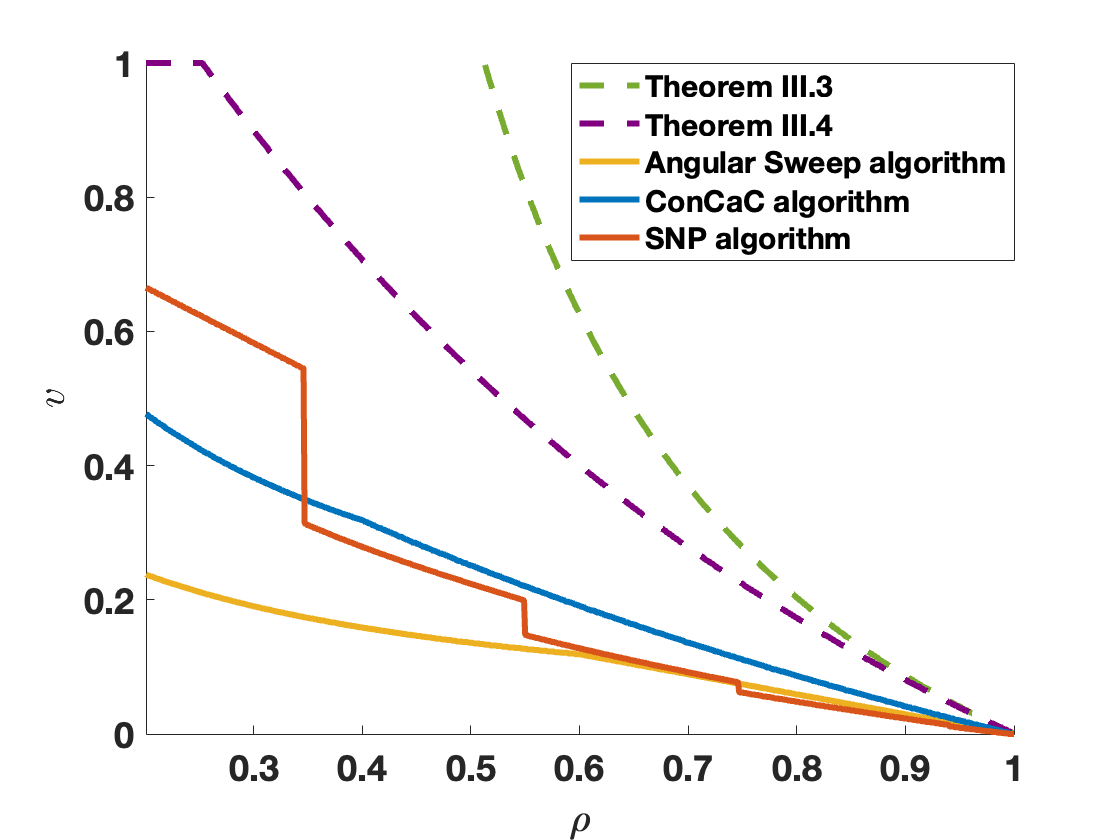}}
\subfigure[$\theta=\pi/2$]{\includegraphics[width=0.3\textwidth]{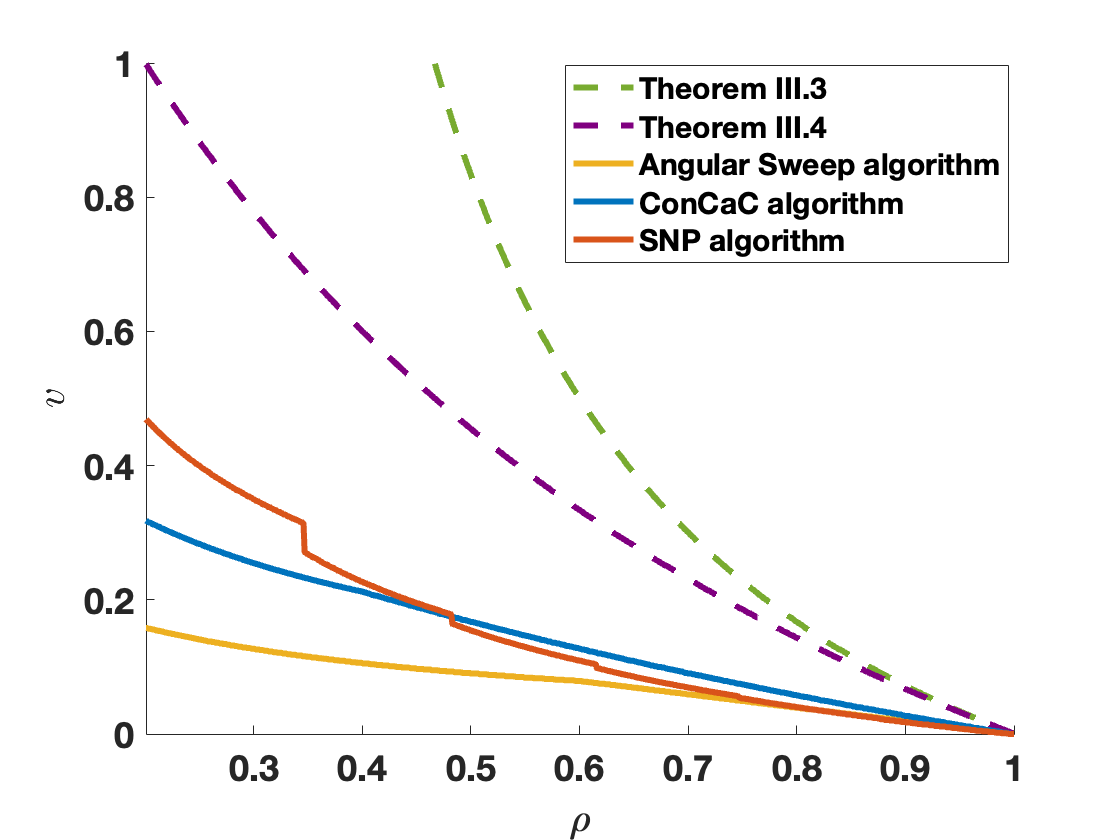}}
\caption{Parameter regime plot for $r=0.2$.}
\label{fig:param_regime_r}
\end{figure*}

\begin{figure*}[t]
\centering
\begin{subfigure}[$r=0.05$]{\includegraphics[width=0.3\textwidth]{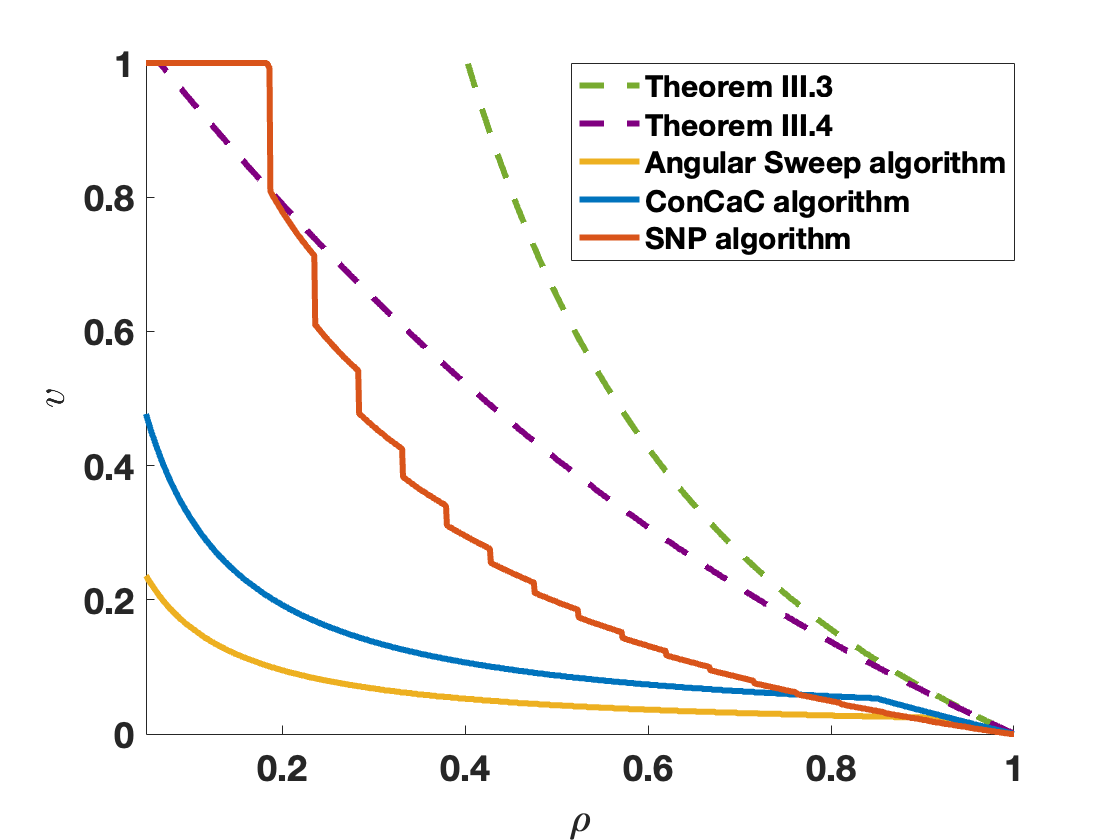}}
\end{subfigure}
\subfigure[$r=0.1$]{\includegraphics[width=0.3\textwidth]{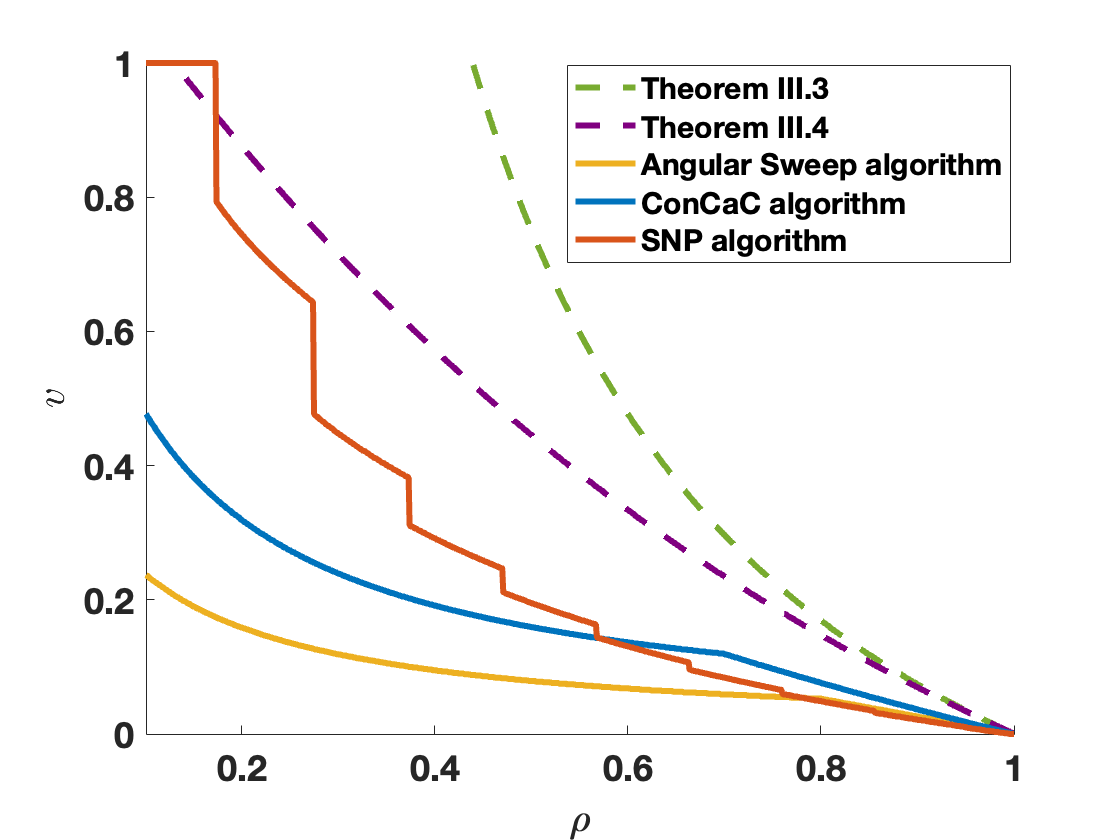}}
\subfigure[$r=0.3$]{\includegraphics[width=0.3\textwidth]{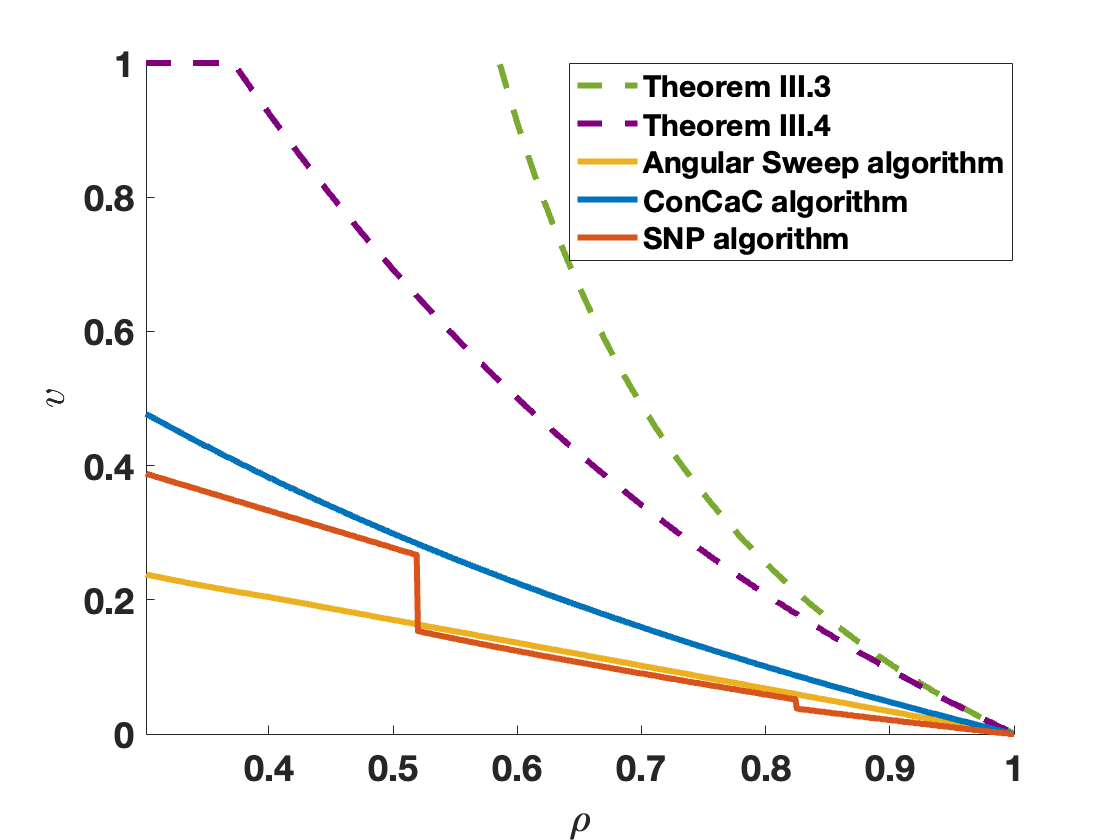}}
\caption{Parameter regime plot for $\theta=\pi/3$.}
\label{fig:param_regime_theta}
\end{figure*}

\begin{lemma}\label{lem:SNP_part2}
Each lost interval is accounted for by the captured intervals of SNP algorithm.
\end{lemma}
\begin{proof}
Note that any realization of SNP can be achieved by the combination of one or more traces as described in the following cases. 
Case (i) $k_i=3$ and $k_l=0 ~\forall l\in \{1,\dots,n_s\}\setminus\{i\}$, 
Case (ii) $0\leq k_i< 3$ and $k_o=2$ and Case (iii) $k_i=0$, $k_o=1$ and $k_{o'}=1, \forall o\in \{1,\dots,n_s\}\setminus\{i\}$ and $\forall o' \in \{ 1,\dots,n_s\}\setminus\{o\}$.
The idea is to identify all of the lost and captured intervals in each case and show that each lost interval is accounted by the captured intervals.

\textbf{Case (i)}: Due to comparison steps \textbf{C1} and \textbf{C2} at time $jD$, the captured intervals $S_i^{j+1}$, $S_i^{j+2}$ and $S_i^{j+3}$ account for all of the lost intervals $S_l^{j+2}$ and $S_l^{j+3}$, $\forall l\in \{ 1,\dots,n_s \}\setminus\{i\}$. There are two sub-cases; sub-case (a) $N_o=N_i$ at time instant $jD$ and sub-case (b), there exists a sector $N_o\neq N_i$ at time instant $jD$ (comparison $\textbf{C1}$) such that $\abs{S_o^{j+2}}<\abs{S_i^{j+1}}$ (comparison $\textbf{C2}$). 
We first consider sub-case (a). Sub-case (a) implies that at time instant $jD$, the total number of intruders in sector $N_i$ is more than in any other sector in the environment. Thus, captured intervals $S_i^{j+1}$, $S_i^{j+2}$ and $S_i^{j+3}$ account for all of the lost intervals $S_l^{j+2}$ and $S_l^{j+3},~\forall l\neq i$.
In sub-case (b), we account for lost intervals $S_l^{j+2}$, $S_l^{j+3}$, $\forall l\in\{1,\dots,n_s\}\setminus\{i,o\}$ and $S_o^{j+2}$, $S_o^{j+3}$, separately. Lost intervals $S_l^{j+2}$ and $S_l^{j+3}$ are accounted for because $\abs{S_l^{j+2}}+\abs{S_l^{j+3}}\leq \abs{S_{i}^{j+1}}+\abs{S_{i}^{j+2}}+\abs{S_{i}^{j+3}}$ or equivalently $\eta_i^l\leq\eta_i^i$  (comparison $\mathbf{C1})$. Now it remains to account for lost intervals $S_o^{j+2}$ and $S_o^{j+3}$. Observe that if there exists a sector $N_o\neq N_i$ at time instant $jD$ such that $\abs{S_o^{j+2}}<\abs{S_i^{j+1}}$, then there cannot exist the same $N_o$ at time instant $(j+1)D$ (from comparison $\textbf{C1}$). Thus, even if $N_o\neq N_i$ exists, then the lost interval $S_o^{j+2}$ is accounted by $S_i^{j+1}$ as $\abs{S_o^{j+2}}<\abs{S_i^{j+1}}$ (comparison $\mathbf{C2}$). Since, at time $(j+1)D$, sector $N_{o}$ cannot be selected again, it follows that $\eta_i^o<\eta_i^i$ at time $(j+1)D$ and thus, $S_o^{j+3}$ is accounted for.

\textbf{Case (ii)}: 
To account for the lost intervals $S_{l}^{j+k_{i}}$ and $S_{l}^{j+1+k_{i}}, \forall l\in \{1,\dots,n_s\}\setminus\{i\}$, from comparison \textbf{C1} and \textbf{C2} at time $(j+k_i)D$, the vehicle was supposed to capture all $S_i^{j-2+k_i}$, $S_{i}^{j-1+k_i}, \dots$, $S_{i}^{j+1+k_{i}}$ intervals. While the vehicle captured $S_i^{j-2+k_i},\dots,S_{i}^{j+k_{i}}$ intervals, it did not capture $S_{i}^{j+1+k_{i}}$. 
As $\eta_i^o>\eta_i^l$ at time instant $(j+k_i)D$, lost intervals $S_{l}^{j+k_{i}}$ and $S_{l}^{j+1+k_{i}}, \forall l\in \{1,\dots,n_s\}\setminus\{i\}$ are fully accounted for. The remaining lost intervals $S_{i}^{j+1+k_{i}}$, $S_{i}^{j+2+k_{i}}$, $S_{i}^{j+3+k_{i}}$ $S_{l}^{j+2+k_{i}}$, and $S_{l}^{j+3+k_{i}}$ $\forall l \in \{1,\dots,n_s\}\setminus\{o\}$ are fully accounted by the captured intervals $S_o^{j+2+k_{i}}$ and $S_o^{j+3+k_{i}}$ because the conditions $\eta_i^o>\eta_i^i$ and $\eta_i^o>\eta_i^l$ are satisfied at time instant $(j+k_i)D$ (comparison $\textbf{C1}$).

\textbf{Case (iii)}: To account for lost intervals $S_{i}^{j+1}$, $S_{i}^{j+2}$, $S_{i}^{j+3}$, $S_{l}^{j+2}$, and $S_{l}^{j+3}$ $\forall l\in \{1,\dots,n_s\}\setminus\{i,o\}$, the vehicle was supposed to capture $S_{o}^{j+2}$ and $S_{o}^{j+3}$. This follows because at time instant $jD$, $\eta_i^o > \eta_i^i$ (comparison $\textbf{C1}$) and $\abs{S_{o}^{j+2}}\geq \abs{S_{i}^{j+1}}$ (comparison $\textbf{C2}$). The vehicle captured $S_{o}^{j+2}$ which accounts for $S_{i}^{j+1}$ as $\abs{S_{o}^{j+2}}\geq \abs{S_{i}^{j+1}}$. As the vehicle moved to capture $S_{o'}^{j+4}$ at time $(j+2)D$, it implies that $\abs{S_{o'}^{j+4}}\geq \abs{S_{o}^{j+3}}$ (comparison $\textbf{C2}$) and thus, $S_{o}^{j+3}$, $S_{i}^{j+2}$, $S_{i}^{j+3}$, $S_{l}^{j+2}$, and $S_{l}^{j+3}$ are all accounted by the captured interval $\abs{S_{o'}^{j+4}}$. Finally, the lost intervals $\abs{S_{l}^{j+4}},\forall l \in \{1,\dots,n_s\}\setminus\{o'\}$ are accounted for as follows: If the vehicle also captures $S_{o'}^{j+5}$, then lost intervals $S_{l}^{j+4}$ are accounted for by per case (ii) ($k_i=1$). Otherwise (i.e., the vehicle moved to another sector $N_{\tilde{o}},\tilde{o}\neq o$ to capture $S_{\tilde{o}}^{j+6}$), $S_{l}^{j+4}$ is accounted for as per case (iii) as now the lost intervals will be $S_{i}^{j+3}$, $S_{i}^{j+4}$, $S_{i}^{j+5}$, $S_{l}^{j+4}$, and $S_{l}^{j+5}$ $\forall l\in \{1,\dots,n_s\}\setminus\{i,o\}$.

Finally, note that the boundary cases of the first and the last intervals fall into these cases by adding dummy intervals $S_i^0, \forall i \in \{ 1,\dots,n_s \}$ and $S_i^{Y+1}$, where $Y$ denotes the last interval that consists of intruders in any sector. We assume that the vehicle captures all of the dummy intervals. This concludes the proof.
\end{proof}

\begin{theorem}[SNP competitiveness]\label{thm:SNP}
For any problem instance $\mathcal{P}(\theta,\rho,v,r)$ that satisfies $3D\leq\frac{1-\rho}{v}$ and $\tfrac{2}{\rho \cos(\theta_s)}\leq 2D$, SNP is $\frac{3n_s-1}{2}$-competitive, where $
n_s = \lceil{\theta}/{\theta_s}\rceil$, $\theta_s = \arctan(r/\rho)$ and $D$ is defined in \eqref{def:Dist}.
\end{theorem}
\begin{proof}
From Lemma \ref{lem:SNP_part1} and Lemma \ref{lem:SNP_part2} it follows that, for any given trace of SNP algorithm, every two consecutively captured intervals pay for $3n_s-3$ lost intervals and every lost interval is accounted by two consecutive captured intervals, and the claim follows.
\end{proof}



\section{Numerical Visualization and Observations}\label{sec:Results}
We now provide a numerical visualization of the analytic bounds derived in this paper. Figure \ref{fig:param_regime_r} shows the $(\rho,v)$ parameter regime plots for a fixed value of capture radius $r=0.2$ and varying values of $\theta$. Figure \ref{fig:param_regime_theta} shows the $(\rho,v)$ parameter regime plots for a fixed value of $\theta=\frac{\pi}{3}$ and varying values of capture radius $r$.

Since the competitiveness of SNP depends on the number of sectors, observe that the parameter regime of SNP is in \textit{regions}, where each region corresponds to a specific competitiveness. As the capture radius $r$ increases, the number of regions decreases and as $\theta$ increases, the number of regions increases. An important characteristic for SNP is that it can be used to determine the tradeoff between the competitiveness and the target parameter regime for the problem instance.

Figure \ref{fig:param_regime_r} suggests that Algorithm SNP has a relatively small area of utility (below the red and above the blue curve) in the parameter space and completely lies below the purple curve for Theorem \ref{thm:At_best_2}. For $\theta=\pi/4$ and $\theta=\pi/3$, SNP is at best $2.5$-competitive for $\rho\leq 0.5$ and $\rho\leq 0.35$, respectively, and decreases by a factor of $1.5$ with every region. For $\rho>0.5$, ConCaC is more effective than SNP as the curve defined by the conditions for ConCaC is completely above the conditions defined for SNP. For $\theta=\pi/2$, SNP is at best $4$-competitive and increases by a factor of $1.5$.

In Figure \ref{fig:param_regime_theta}, for small values of $r$ ($0.05$ and $0.1$), SNP has a relatively large area of utility.  For $r=0.05$ and $r=0.1$, SNP is at best $2.5$-competitive for $\rho<0.2$ and $\rho<0.17$. This suggests with the smaller the capture radius, SNP can capture equally fast intruders, and it covers a larger area in the parameter regime but at the cost of higher competitive ratio. Interestingly, the curve for SNP extends beyond that of Theorem \ref{thm:At_best_2}. For $r=0.3$, the curve defined by sufficient conditions for SNP is completely below the curve defined by conditions of ConCaC suggesting that SNP is ineffective for high values of $r$.

\section{Conclusion and Future Directions}\label{sec:Conclusion}
This work analyzed the problem wherein a single vehicle, having a finite capture radius $r$, is tasked to defend a perimeter in a conical environment from arbitrary many intruders that arrive in the environment in an arbitrary fashion. We designed and analyzed three algorithms, i.e., Angular Sweep, Conical Compare and Capture, and Stay Near Perimeter algorithms, and established sufficient conditions that guarantee a finite competitive ratio for each algorithm. In particular, we demonstrated that Angular Sweep algorithm is $1$-competitive and Conical Compare and Capture algorithm is $2$ competitive for parameter space beyond that of Angular Sweep algorithm. Moreover, the competitive ratio of Stay Near Perimeter changes with as a function of the parameters ($r,\rho,\theta$) and does not always extend beyond that of Conical Compare and Capture algorithm in specific parameter regimes. Thus, the choice of which algorithm to use depends on the problem parameters and the acceptable bound on competitiveness. We also derived two fundamental limits on achieving a finite competitive ratio by any online algorithm. 

Apart from closing the gap between the curve defined by Theorem \ref{thm:CaC} for algorithm ConCaC and the curve defined by Theorem \ref{thm:At_best_2} as well as the gap between the curve defined by Theorem \ref{thm:SNP} for algorithm SNP and Theorem \ref{thm:no_c}, key future directions include a cooperative multi-vehicle scenario with communication and energy constraints. 

\bibliographystyle{IEEEtran}
\bibliography{references}
\end{document}